\documentclass{llncs}

\usepackage[boxruled,vlined,linesnumbered,noend]{algorithm2e}
\usepackage{todonotes}
\usepackage{graphicx}
\usepackage{amsmath}
\usepackage{amssymb}
\usepackage{float}
\usepackage{fullpage}

\newcommand{\branchvector}[1]{{\color{red}{$(#1)$}}}
\newcommand{\FPT}{\textsf{FPT}}
\newcommand{\WT}{\textsf{W[2]}}

\newcommand{\WO}{\textsf{W[1]}}

\newcommand{\NP}{\textsf{NP}}

\newcommand{\hdel}{{\sc $\mathcal{F}$-Deletion}}
\newcommand{\alphahdel}{{\sc $\alpha$-Simultaneous $\mathcal{F}$-Deletion}}
\newcommand{\nfvs}{{\sc $\OO(n)$-SimFVS}}

\newcommand{\logfvs}{{\sc $\OO(\log n)$-SimFVS}}

\newcommand{\alphafvs}{{\sc $\alpha$-SimFVS}}
\newcommand{\alphafvsfull}{{\sc $\alpha$-Simultaneous Feedback Vertex Set}}
\newcommand{\alphafvssol}{{$\alpha$-simultaneous feedback vertex set}}
\newcommand{\alphafvssolshort}{{$\alpha$-simfvs}}
\newcommand{\fvssolshort}{{simfvs}}
\newcommand{\disalphafvs}{{\sc Disjoint $\alpha$-SimFVS}}

\newcommand{\twofvs}{{\sc $2$-SimFVS}}
\newcommand{\twofvsfull}{{\sc $2$-Simultaneous Feedback Vertex Set}}

\newcommand{\cordate}{cordate}
\newcommand{\onetoalpha}{\{1,2,\dots,\alpha \}}
\newcommand{\onetoalphaNOi}{\{1,2,\dots,\alpha \} \setminus \{i\}}
\newcommand{\degbound}{$\OO(k^2)$}
\newcommand{\degboundT}{$\OO(k^3)$}

\newcommand{\pfd}{{\sc Planar $\cal F$-Deletion}}

\newcommand{\OO}{\mathcal{O}}

\newcommand{\defparproblem}[4]{
 \vspace{3mm}
\noindent\fbox{
  \begin{minipage}{.95\textwidth}
  \begin{tabular*}{\textwidth}{@{\extracolsep{\fill}}lr} \textsc{#1}  & {\bf{Parameter:}} #3 \\ \end{tabular*}
  {\bf{Input:}} #2  \\
  {\bf{Question:}} #4
  \end{minipage}
  }
  \vspace{2mm}
}

\title{Simultaneous Feedback Vertex Set: A Parameterized Perspective}
\author
{
    Akanksha Agrawal\inst{1}
    \and Daniel Lokshtanov\inst{1}
    \and Amer E. Mouawad\inst{1}
    \and Saket Saurabh\inst{1,2}
}
\institute
{
    University of Bergen\\
    Bergen, Norway.\\
    \email{\{akanksha.agrawal|daniel.lokshtanov|a.mouawad\}@uib.no}
    \and
    Institute of Mathematical Sciences\\
    Chennai, India.\\
    \email{saket@imsc.res.in}\\
}

\begin{document}

\maketitle

\begin{abstract}
	Given a family of graphs $\mathcal{F}$, a graph $G$, and a positive integer $k$,
the \hdel\ problem asks whether we can delete at most $k$ vertices from $G$
to obtain a graph in $\mathcal{F}$. \hdel\ generalizes many classical
graph problems such as {\sc Vertex Cover}, {\sc Feedback Vertex Set}, and {\sc Odd Cycle Transversal}.
A graph $G  = (V, \cup_{i=1}^{\alpha} E_{i})$,  where the edge set of $G$ is partitioned into $\alpha$ color classes,
is called an $\alpha$-edge-colored graph.
%A graph $G$ with its edge set $E(G)$ being partitioned into  $\cup_{i=1}^{\alpha} E_{i}$ is called {\em $\alpha$-edge-colored graph}. In other words  the edge set of $G$ is partitioned into $\alpha$ color classes.
% we say $G$ is an
%$\alpha$-edge-colored graph, i.e. the edge set of $G$ can be partitioned into $\alpha$ color classes.
A natural extension of the \hdel\ problem to edge-colored graphs is
the \alphahdel\ problem. In the latter problem, we are given an
$\alpha$-edge-colored graph $G$ and the goal is to find a set $S$ of at most $k$ vertices
such that each graph $G_i \setminus S$, where $G_i = (V, E_i)$ and $1 \leq i \leq \alpha$, is in $\mathcal{F}$.
In this work, we study \alphahdel\ for $\mathcal{F}$ being  the
family of forests. In other words, we focus on the \alphafvsfull\ (\alphafvs) problem.
%Although very tightly related to the classical {\sc Feedback Vertex Set} problem, we show
%that there are subtle and interesting differences and we believe that studying
%\alphahdel\ for other families, e.g. bipartite or
%planar graphs, might constitute a fruitful research direction.
Algorithmically, we show that, like its classical counterpart, \alphafvs\ parameterized
by $k$ is fixed-parameter tractable (\FPT) and admits a polynomial kernel, for any fixed constant $\alpha$.
In particular, we give an algorithm running in $2^ {\OO(\alpha k)}n^{\OO(1)}$ time and a kernel with
$\OO(\alpha k^{3(\alpha + 1)})$ vertices. The running time of our algorithm implies that \alphafvs\  is \FPT\ even when
$\alpha \in o(\log n)$.  We complement this positive result by showing that for $\alpha \in \OO(\log n)$, where
$n$ is the number of vertices in the input graph, \alphafvs\ becomes \WO-hard. Our positive results
answer one of the open problems posed by Cai and Ye (MFCS 2014).
%Then, we consider the dependence between $\alpha$ and both the size of our kernel
%and the running time of our algorithm. We show that
%even for $\alpha \in \OO(\log n)$, where $n$ is the number of vertices
%in the input graph, \alphafvs\ becomes \WO-hard. We prove \WO-hardness via
%a new problem of independent interest. 
\end{abstract}

\section{Introduction}
In graph theory, one can define a general family of problems as follows.
Let $\cal F$ be a collection of graphs. Given an undirected graph $G$ and a
positive integer $k$, is it possible to perform at most $k$ edit operations to $G$ so
that the resulting graph does not contain a graph from $\cal F$? Here one can define
edit operations as either vertex/edge deletions, edge additions, or edge contractions.
Such problems constitute a large fraction of problems considered under the parameterized complexity framework.
When edit operations are restricted to vertex deletions this corresponds to the \hdel\ problem, which
generalizes classical graph problems such as \textsc{Vertex Cover}~\cite{Chen:2010:IUB:1850840.1850901},
\textsc{Feedback Vertex Set}~\cite{Chen20081188,Cygan:2011:SCP:2082752.2082943,Kociumaka2014556},
\textsc{Vertex Planarization}~\cite{planardel}, {\sc Odd Cycle Transversal}~\cite{10.1109/FOCS.2012.46,Lokshtanov:2014:FPA:2685353.2566616},
{\sc Interval Vertex Deletion}~\cite{Cao:2015:IDF:2721890.2629595},
{\sc Chordal Vertex Deletion}~\cite{chordaldel}, and \pfd~\cite{conf:focs:FominLMS12,kimicalp2013}.
The topic of this paper is a generalization of \hdel\ problems to ``edge-colored graphs''.
In particular, we do a case study of an edge-colored version of the
classical \textsc{Feedback Vertex Set} problem~\cite{Garey:1979:CIG:578533}.
%\todo{add citations here}

A graph $G  = (V, \cup_{i=1}^{\alpha} E_{i})$,  where the edge set of $G$ is partitioned into $\alpha$ color classes,
is called an {\em $\alpha$-edge-colored graph}. As stated  by Cai and
Ye~\cite{caiye2014}, ``edge-colored graphs are fundamental in graph theory
and have been extensively studied in the literature, especially for
alternating cycles, monochromatic sub-graphs, heterchromatic subgraphs, and partitions''.
A natural extension of the \hdel\ problem to edge-colored graphs is
the \alphahdel\ problem. In the latter problem, we are given an
$\alpha$-edge-colored graph $G$ and the goal is to find a set $S$ of at most $k$ vertices
such that each graph $G_i \setminus S$, where $G_i = (V, E_i)$ and $1 \leq i \leq \alpha$, is in $\mathcal{F}$.
 %A graph $G$ with its edge set $E(G)$ being partitioned into  $\cup_{i=1}^{\alpha} E_{i}$ is called {\em $\alpha$-edge-colored graph}. In other words  the edge set of $G$ is partitioned into $\alpha$ color classes.
% we say $G$ is an
%$\alpha$-edge-colored graph, i.e. the edge set of $G$ can be partitioned into $\alpha$ color classes.
Cai and Ye~\cite{caiye2014} studied several problems
restricted to $2$-edge-colored graphs, where edges are colored either red or blue.
In particular, they consider the {\sc Dually Connected Induced Subgraph} problem,
i.e. find a set $S$ of $k$ vertices in $G$ such that both induced graphs
$G_{\text{red}}[S]$ and $G_{\text{blue}}[S]$ are connected,
and the {\sc Dual Separator} problem, i.e.
delete a set $S$ of at most $k$ vertices to simultaneously disconnect
the red and blue graphs of $G$.
They show, among other results, that {\sc Dual Separator} is \NP-complete
and {\sc Dually Connected Induced Subgraph} is \WO-hard even
when both $G_{\text{red}}$ and $G_{\text{blue}}$ are trees.
On the positive side, they prove that {\sc Dually Connected Induced Subgraph} is solvable
in time polynomial in the input size when $G$ is a complete graph.
One of the open problems they state is to determine the parameterized complexity
of \alphahdel\ for $\alpha = 2$ and $\mathcal{F}$
the family of forests, bipartite graphs, chordal graphs, or planar graphs.
The focus in this work is on one of those problems, namely \alphafvsfull --- an  interesting, and well-motivated~\cite{app1,caiye2014,app2}, generalization of {\sc Feedback Vertex Set} on edge-colored graphs.

A \emph{feedback vertex set} is a subset $S$ of vertices such that $G \setminus S$ is a forest.
For an $\alpha$-colored graph $G$, an {\em \alphafvssol\ } (or {\em \alphafvssolshort} for short) is a subset $S$ of vertices
such that $G_i \setminus S$ is a forest for each $1 \leq i \leq \alpha$. The \alphafvsfull\ is stated formally as follows.
%We study the $\alpha${\emph-colored feedback vertex set} which is the following.

\defparproblem{\alphafvsfull\  (\alphafvs)  }{$(G,k)$, where $G$ is an undirected $\alpha$-colored graph and $k$ is a positive integer}{$k$}
{Is there a subset $S\subseteq V(G)$ of size at most $k$ such that for $1 \leq i \leq \alpha$, $G_i\setminus S$ is a forest?}

\noindent
%%classical FVS.
Given a graph $G = (V,E)$ and a positive integer $k$, the classical {\sc Feedback Vertex Set (FVS)} problem asks whether
there exists a set $S$ of at most $k$ vertices in $G$ such that the graph induced on $V(G) \setminus S$
is acyclic. In other words, the goal is to find a set of at most $k$ vertices that intersects all cycles in $G$.
{\sc FVS} is a classical \NP-complete~\cite{Garey:1979:CIG:578533} problem with
numerous applications and is by now very well understood
from both the classical and parameterized complexity~\cite{DF97} view points. For instance, the problem
admits a $2$-approximation algorithm~\cite{2-approx-fvs-bafna}, an exact (non-parameterized) algorithm
running in $\OO^\star(1.736^n )$ time~\cite{exactfvs}, a deterministic algorithm
running in $\OO^\star(3.619^k)$ time~\cite{Kociumaka2014556}, a randomized algorithm running
in $\OO^\star(3^k )$ time~\cite{Cygan:2011:SCP:2082752.2082943},
and a kernel on $\OO(k^2)$ vertices~\cite{Thomasse:2010:KKF:1721837.1721848} (see Section~\ref{prelim} for definitions). We use the $\OO^\star$ notation to describe the running times of our algorithms.
Given $f: \mathbb{N} \rightarrow \mathbb{N}$, we define
$\OO^\star(f(n))$ to be $\OO(f(n) \cdot p(n))$, where $p(\cdot)$ is some
polynomial function. That is, the $\OO^\star$ notation suppresses polynomial factors in the running-time expression.
%\todo[inline]{Define $\OO^\star$}

%%new FVS1 and new FVS2
%A natural generalization of the {\sc Feedback Vertex Set} problem
%is the \hdel\ problem.
%Given a family of graphs $\mathcal{H}$, a graph $G$, and a positive integer $k$,
%the {\sc $\mathcal{H}$-Deletion} problem asks whether we can delete at most $k$ vertices from $G$
%to obtain a graph in $\mathcal{H}$. When $\mathcal{H}$ is the set of
%all forests then {\sc $\mathcal{H}$-Deletion} is exactly {\sc Feedback Vertex Set}.
%The {\sc $\mathcal{H}$-Deletion} problem has also received considerable attention in
%the literature and the (parameterized) complexity status has been resolved for many
%graph families~\cite{saurabh-book}.

%In this work, we consider another interesting, and well-motivated~\cite{app1,caiye2014,app2}, generalization of {\sc Feedback Vertex Set} on edge-colored graphs.
%For a graph $G = (V, \cup_{i=1}^{\alpha} E_{i})$, we say $G$ is an
%$\alpha$-edge-colored graph, i.e. the edge set of $G$ can be partitioned into $\alpha$ color classes.
%Given an $\alpha$-edge-colored graph $G$ and an integer $k$,
%the \alphafvsfull\ (\alphafvs) problem asks for a set $S$ of at most $k$ vertices
%such that each graph $G_i \setminus S$, where $G_i = (V, E_i)$ and $1 \leq i \leq \alpha$, is acyclic.
%Alternatively, one can consider the even more general \alphahdel\ problem,
%where given an $\alpha$-edge-colored graph $G$ the goal is to find a set $S$ of at most $k$ vertices
%such that each graph $G_i \setminus S$, $1 \leq i \leq \alpha$, is in $\mathcal{H}$.

\vspace*{2mm}
\noindent
\textbf{Our results and methods.} We show that, like its classical counterpart, \alphafvs\ parameterized
by $k$ is \FPT\  and admits a polynomial kernel, for any fixed constant $\alpha$. In particular, we obtain the
following results.
\begin{itemize}
\item An \FPT\ algorithm running in $\OO^\star(23^{\alpha k})$ time. For the special case of $\alpha=2$, we give a faster
algorithm running in $\OO^\star(81^{k})$ time.
\item For constant $\alpha$, we obtain a kernel with $\OO(\alpha k^{3(\alpha + 1)})$ vertices.
\item The running time of our algorithm implies that \alphafvs\  is \FPT\ even when
$\alpha \in o(\log n)$.  We complement this positive result by showing that for $\alpha \in \OO(\log n)$, where
$n$ is the number of vertices in the input graph, \alphafvs\ becomes \WO-hard.
\end{itemize}

Our algorithms and kernel build on the tools and methods
developed for {\sc FVS}~\cite{saurabh-book}. However, we need to develop both new
branching rules as well as a new reduction rules. The main reason why our results do not follow directly from
earlier work on {\sc FVS} is the following.
Many (if not all) parameterized algorithms, as well as kernelization algorithms, developed
for the {\sc FVS} problem~\cite{saurabh-book} exploit the fact that vertices of degree two or less
in the input graph are, in some sense, irrelevant. In other words, vertices of degree one or zero
cannot participate in any cycle and every cycle containing any
degree-two vertex must contain both of its neighbors.
%the number of edge-disjoint cycles
%any degree-two vertex can participate in is exactly one,
Hence, if this degree-two vertex is part of a feedback vertex set then it can be replaced by either one of its neighbors.
Unfortunately (or fortunately for us), this property does not hold for the
\alphafvs\ problem, even on graphs where edges are bicolored either red or blue.
For instance, if a vertex is incident to two red edges and two blue edges, it might in fact be participating
in two distinct cycles. Hence, it is not possible to neglect (or shortcut) this vertex in neither
$G_{\text{red}}$ nor $G_{\text{blue}}$. As we shall see, most of the new algorithmic techniques that we present
deal with vertices of exactly this type.
Although very tightly related to one another, we show that there are subtle and interesting differences separating
the {\sc FVS} problem from the \alphafvs\ problem, even for $\alpha = 2$.
For this reason, we also believe that studying \alphahdel\ for different families of graphs $\mathcal{F}$,
e.g. bipartite, chordal, or planar graphs, might reveal some new insights about the classical underlying problems.

%The rest of this paper is organized as follows.
In Section~\ref{sec-algo}, we present an algorithm solving the \alphafvs\ problem,
parameterized by solution size $k$, in $\OO^\star(23^{\alpha k})$ time. Our algorithm follows the iterative compression
paradigm introduced by Reed et al.~\cite{Reed2004299} combined with new reduction and branching rules. Our main new
branching rule can be described as follows: Given a maximal degree-two path in some $G_i$, $1 \leq i \leq \alpha$,
we branch depending on whether there is a vertex from this path participating in an \alphafvssol\ or not.
In the branch where we guess that a solution contains a vertex from this path, we construct a color $i$ cycle
which is isolated from the rest of the graph. In the other branch,
we are able to follow known strategies by ``simulating'' the classical {\sc FVS} problem.
Observe that we can never have more than $k$ isolated cycles of the same
color. Hence, by incorporating this fact into our measure we are guaranteed to make ``progress'' in both branches.
For the base case, each $G_i$ is a disjoint union of cycles (though not $G$) and to
find an \alphafvssol\ for $G$ we cast the remaining problem as an instance of
{\sc Hitting Set} parameterized by the size of the family.
For $\alpha=2$, we can instead use an algorithm for finding maximum matchings in an auxiliary graph.
Using this fact we give a faster, $\OO^\star(81^{k})$ time, algorithm for the case $\alpha = 2$.
In Section~\ref{sec-kernel}, we tackle the question of kernelization and present a polynomial
kernel for the problem, for constant $\alpha$. Our kernel has $\OO(\alpha k^{3(\alpha+1)})$ vertices
and requires new insights into the possible structures induced by those special vertices discussed above. In particular, we enumerate all maximal degree-two paths in each $G_i$ after deleting a feedback vertex set in $G_i$ and study how such paths interact with each other.
Using marking techniques, we are able to ``unwind''  long degree-two paths by making a private copy of each
unmarked vertices for each color class. This unwinding leads to ``normal'' degree-two
paths on which classical reduction rules can be applied and hence we obtain the desired kernel.

Finally, we consider the dependence between $\alpha$ and both the size of our kernel
and the running time of our algorithm in Section~\ref{sec-hardness}. We show that
even for $\alpha \in \OO(\log n)$, where $n$ is the number of vertices
in the input graph, \alphafvs\ becomes \WO-hard. We show
hardness via a new problem of independent interest which we denote by {\sc $\alpha$-Partitioned Hitting Set}.
The input to this problem consists of a tuple $(\mathcal{U}, \mathcal{F} = \mathcal{F}_1 \cup \ldots \cup \mathcal{F}_{\alpha}, k)$,
where $\mathcal{F}_{i}$, $1 \leq i \leq \alpha$,
is a collection of subsets of the finite universe $\mathcal{U}$, $k$ is a positive integer, and
all the sets within a family $\mathcal{F}_{i}$, $1 \leq i \leq \alpha$, are pairwise disjoint.
The goal is to determine whether there exists a subset $X$ of $\mathcal{U}$ of cardinality at most $k$ such that
for every $f \in \mathcal{F} = \mathcal{F}_1 \cup \ldots \cup \mathcal{F}_{\alpha}$, $f \cap X$ is nonempty.
We show that {\sc $\OO(\log |\mathcal{U}||\mathcal{F}|)$-Partitioned Hitting Set} is
\WO-hard via a reduction from {\sc Partitioned Subgraph Isomorphism}
and we show that {\sc $\OO(\log n)$-SimFVS} is \WO-hard via a
reduction from {\sc $\OO(\log |\mathcal{U}||\mathcal{F}|)$-Partitioned Hitting Set}.
Along the way, we also show, using a somewhat simpler reduction from {\sc Hitting Set},
that {\sc $\OO(n)$-SimFVS} is \WT-hard.
%\todo{May be say a line about $\OO(n)$-Sim-FVS that it is W[2]-hard.}

\section{Preliminaries}\label{prelim}
We start with some basic definitions and introduce
terminology from graph theory and algorithms.
We also establish some of the notation that will be used throughout.

%We will use the $\OO^\star$ notation to describe the running times of our algorithms.
%Given $f: \mathbb{N} \rightarrow \mathbb{N}$, we define
%$\OO^\star(f(n))$ to be $\OO(f(n) \cdot p(n))$, where $p(\cdot)$ is some
%polynomial function. That is, the $\OO^\star$ notation suppresses polynomial factors in the running-time expression.

\vspace*{2mm}
\noindent
\textbf{Graphs.} For a graph $G$, by $V(G)$ and $E(G)$ we denote its vertex set and edge set, respectively.
We only consider finite graphs possibly having loops and multi-edges.
In the following, let $G$ be a graph and let $H$ be a subgraph of $G$. By $d_{H}(v)$, we denote the
degree of vertex $v$ in $H$. For any non-empty subset $W \subseteq V(G)$, the
subgraph of $G$ induced by $W$ is denoted by $G[W]$; its vertex set is $W$ and its
edge set consists of all those edges of $E$ with both endpoints in $W$.
For $W \subseteq V(G)$, by $G \setminus W$ we denote the graph obtained by
deleting the vertices in $W$ and all edges which are incident to at least one vertex in $W$.

A \emph{path} in a graph is a sequence of distinct vertices $v_0, v_1, \ldots, v_k$
such that $(v_i,v_{i+1})$ is an edge for all $0 \leq i < k$.
A \emph{cycle} in a graph is a sequence of distinct vertices $v_0, v_1, \ldots, v_k$
such that $(v_i,v_{(i+1)\mod k})$ is an edge for all $0 \leq i \leq k$. We note that
both a double edge and a loop are cycles. We also use the convention
that a loop at a vertex $v$ contributes $2$ to the degree of $v$.

An edge $\alpha$-colored graph is a graph $G=(V,\cup_{i=1}^{\alpha} E_{i})$.
We call $G_i$ the color $i$ (or $i$-color) graph of $G$, where $G_i=(V,E_i)$.
%For the sake of clarity, we will sometimes assume that each $G_i$ has its own copy of the vertex set. Deleting a vertex from $G_i$ does not delete a vertex from $G_j$, for $i \neq j$, unless $v$ is deleted from all $G_j$, for $1 \leq j \leq \alpha$. Furthermore, deleting a vertex $v$ from $V(G)$, deletes $v$ in each $G_i$. 
For notational convenience we sometimes denote an $\alpha$-colored graph as $G=(V,E_1,E_2,...,E_{\alpha})$.
For an $\alpha$-colored graph $G$, the \emph{total degree} of a vertex
$v$ is $\sum_{i=1}^{\alpha} d_{G_i}(v)$. By color $i$ edge (or $i$-color edge) we
refer to an edge in $E_i$, for $1 \leq i \leq \alpha$. A vertex $v\in V(G)$ is
said to have a color $i$ neighbor if there is an edge $(v,u)$ in $E_i$, furthermore $u$ is a color $i$ neighbor of $v$.
We say a path or a cycle in $G$ is {\em monochromatic} if all the edges on the path or cycle have the same color.
Given a vertex $v \in V(G)$, a {\em $v$-flower} of order $k$ is a
set of $k$ cycles in $G$ whose pairwise intersection is exactly $\{v\}$. If all cycles
in a $v$-flower are monochromatic then we have a {\em monochromatic $v$-flower}.
An $\alpha$-colored graph $G=(V,E_1,E_2, \cdots , E_\alpha)$ is an
\emph{$\alpha$-forest} if each $G_i$ is a forest, for $1 \leq i \leq \alpha$. We refer the reader to~\cite{diestel-book} for details on standard graph theoretic notation and terminology we use in the paper.

%A \emph{feedback vertex set} is a subset $S$ of vertices such that $G \setminus S$ is a forest.
%For an $\alpha$-colored graph $G$, an {\em \alphafvssol\ } (or {\em \alphafvssolshort\ } for short) is a subset $S$ of vertices
%such that $G_i \setminus S$ is a forest for each $1 \leq i \leq \alpha$.
%%We study the $\alpha${\emph-colored feedback vertex set} which is the following.
%
%\defparproblem{\alphafvsfull\  (\alphafvs)  }{$(G,k)$, where $G$ is an undirected $\alpha$-colored graph and $k$ is a positive integer}{$k$}
%{Is there a subset $S\subseteq V(G)$ of size at most $k$ such that for $1 \leq i \leq \alpha$, $G_i\setminus S$ is a forest?}

\vspace*{2mm}
\noindent
\textbf{Parameterized Complexity.} A parameterized problem $\Pi$ is a subset of $\Gamma^{*}\times\mathbb{N}$, where $\Gamma$ is a finite alphabet. An instance of a parameterized problem is a tuple $(x,k)$, where $x$ is a classical problem instance, and $k$ is called the parameter. A central notion in parameterized complexity is {\em fixed-parameter tractability (FPT)} which means, for a given instance $(x,k)$, decidability in time $f(k)\cdot p(|x|)$, where $f$ is an arbitrary function of $k$ and $p$ is a polynomial in the input size.

\vspace*{2mm}
\noindent
\textbf{Kernelization.} A kernelization algorithm for a parameterized problem   $\Pi\subseteq \Gamma^{*}\times \mathbb{N}$ is an algorithm that, given $(x,k)\in \Gamma^{*}\times \mathbb{N} $, outputs, in time polynomial in $|x|+k$, a pair $(x',k')\in \Gamma^{*}\times \mathbb{N}$ such that (a) $(x,k)\in \Pi$ if and only if  $(x',k')\in \Pi$ and (b) $|x'|,k'\leq g(k)$, where $g$ is some computable function. The output instance $x'$ is called the kernel, and the function $g$ is referred to as the size of the kernel. If $g(k)=k^{\OO(1)}$ (resp. $g(k)=\OO(k)$) then we say that $\Pi$ admits a polynomial (resp. linear) kernel.

\section{\FPT~algorithm for \alphafvsfull}\label{sec-algo}
We give an algorithm for the \alphafvs~problem using the method of iterative compression~\cite{Reed2004299,saurabh-book}.
We only describe the algorithm for the disjoint version of the problem.
The existence of an algorithm running in $c^k\cdot n^{\OO(1)}$ time for the disjoint variant implies that
\alphafvs~can be solved in time $(1+c)^k \cdot n^{\OO(1)}$~\cite{saurabh-book}. In the \disalphafvs~problem, we are given an $\alpha$-colored graph $G$ $=$ $(V$,$E_1$,$E_2$, $\dots$, $E_{\alpha})$,
an integer $k$, and an \alphafvssolshort~$W$ in $G$ of size $k+1$. The objective is to
find an \alphafvssolshort~$X \subseteq V(G)\setminus W$ of size at most $k$, or correctly
conclude the non-existence of such an \alphafvssolshort.

\subsection{Algorithm for \disalphafvs}
Let $(G=(V,E_1,E_2, \dots, E_{\alpha}),W,k)$ be an instance of \disalphafvs~and
let $F=G\setminus W$. We start with some simple reduction rules that clean up the graph.
Whenever some reduction rule applies, we apply the lowest-numbered applicable rule.

 \begin{algorithm}[!h]
 \SetAlgoLined
 %\DontPrintSemicolon
 \KwIn{$G=(V,E_1,E_2,\dots, E_{\alpha})$, $W$, $k$, and $\mathfrak{C}=\{\mathcal{C}_1, \dots, \mathcal{C}_{\alpha}\}$}

\KwOut{\textsc{YES} if $G$ has an \alphafvssolshort~$S\subseteq V(G)\setminus W$ of size at most $k$, \textsc{NO} otherwise.}
Apply \textsc{\alphafvs~R.1} to \textsc{\alphafvs~R.5} exhaustively\;

\uIf{$k < 0$ or for any $i \in \onetoalpha$, $\lvert \mathcal{C}_i \rvert > k$}{\KwRet{\textsc{NO}}}

%\uIf{}{\KwRet{\textsc{NO}}}

%\tcp{Branching Rules}
\While{for some $i \in \onetoalpha$, $G_i[F_i\cup W_i]$ is not a forest}{
find a \cordate~vertex $v_c$ of highest index in some tree of $F_i$\;
Let $u_c,w_c$ be the vertices in tree $T^i_{v_c}$ with a neighbor $u,w$ respectively in $W_i$\;
Also let $P=u_c,x_1,\dots, x_t,v_c$ and $P'=v_c,y_1,\dots, y_{t'},w_c$ be the path in $F_i$ from $u_c$ to $v_c$ and $v_c$ to $w_c$ respectively\;

$\mathcal{G}_1=(G\setminus\{v_{c}\},W,k-1, \mathfrak{C})$, Add $\mathcal{G}_1$ to $\mathfrak{G}$\;

\uIf{$V'=V(P)\setminus \{v_c\} \neq \emptyset$}{
$\mathcal{C}_i=\mathcal{C}_i \cup \{(u_c,x_1,\dots,x_t)\}$\;
%$\mathfrak{C}'=(\mathfrak{C}\setminus \{\mathcal{C}_i\}) \cup \{\mathcal{C}'_i\}$\;
$\mathcal{G}_2=(G\setminus V',W,k-1, \mathfrak{C})$, Add $\mathcal{G}_2$ to $\mathfrak{G}$\;
}

\uIf{$V'=V(P')\setminus \{v_c\} \neq \emptyset$}{
$\mathcal{C}_i=\mathcal{C}_i \cup \{(y_1,\dots,y_{t'},w_c)\}$\;
$\mathcal{G}_3=(G\setminus V',W,k-1, \mathfrak{C})$, Add $\mathcal{G}_3$ to $\mathfrak{G}$\;
}

\eIf{$u,w$ are in the same component of $W_i$}{
\KwRet{$\bigvee_{\mathcal{G}\in \mathfrak{G}}$\disalphafvs$(\mathcal{G})$}
}
{ \KwRet{$( \bigvee_{\mathcal{G}\in \mathfrak{G}}$\disalphafvs$(\mathcal{G}))$ $\vee$ \disalphafvs$(G\setminus  (V(P)\cup V(P')), W \cup V(P) \cup V(P'),k,\mathfrak{C})$ } }
}
\tcp{Solve the remaining instance using the hitting set problem}
For $i \in \onetoalpha$ let $V(\mathcal{C}_i)=\cup_{C \in \mathcal{C}_i}V(C)$, $\mathcal{U}=\cup_{i \in \onetoalpha}V(\mathcal{C}_i)$\;
$\mathcal{F}=\cup_{i \in \onetoalpha} \mathcal{C}_i$\;
Find a hitting set $S=$ \textsc{Hitting Set}($\mathcal{F},\mathcal{U}$)\;

\uIf{$\lvert S \rvert \leq k$}{\KwRet{\textsc{YES}}}

%\eIf{$\lvert S \rvert \leq k$}{\KwRet{\textsc{YES}}}
\KwRet{\textsc{NO}}

%\tcp{Reduction Rules}
%

%
%\uIf{$G$ has a vertex $v$ of total degree $0$}{\KwRet{\disalphafvs($G\setminus \{v\}, W, k,\mathcal{C}$)}}
%
%\uIf{for any $i \in \onetoalpha$, $G_i$ has a vertex $v$, with the only neighbor $u$ in $G_i$}{\KwRet{\disalphafvs($G=(V,E_1,E_2,\dots, E_i\setminus \{(v,u)\},\dots, E_{\alpha}), W, k,\mathcal{C}$)}}
%
%\uIf{for any $i \in \onetoalpha$, $G_i$ has a vertex $v$, with only two neighbors $u,w$ in $G_i$ and total degree of $v$ is two}
%{Let $E'_i=(E_i\setminus \{(v,u),(v,w)\}) \cup \{(u,w)\}$\;
%\KwRet{\disalphafvs($G=(V,E_1,E_2,\dots, E'_i,\dots, E_{\alpha}), W, k,\mathcal{C}$)}}
%
%\uIf{for any $i \in \onetoalpha$, $E_i$ has an edge $(u,v)$ with multiplicity more than $2$}
%{Let $E'_i$ be the edge set obtained after reducing multiplicity of $(u,v)$ in $E_i$ to $2$\;
%\KwRet{\disalphafvs($G=(V,E_1,E_2,\dots, E'_i,\dots, E_{\alpha}), W, k,\mathcal{C}$)}}
%
%\uIf{$G$ has a vertex $v$ with a self loop}{\KwRet{\disalphafvs($G\setminus \{v\}, W, k-1,\mathcal{C}$)}}
%For $i\in \onetoalpha$, if there is a vertex $v$ with $u$ as its only neighbor in $G_i$ then, delete $(v,u)$ from $E_i$\;
% \uIf{$k \geq 0$ and $G$ is empty}{\KwRet{\YES}}
% \uIf{$k \geq 0$ and $G$ is empty}{\KwRet{\YES}}
% \uIf{$k = 0$ and $G$ is not empty}{\KwRet{\NO}}
 \caption{\disalphafvs}
 \label{overall_algorithm}
 \end{algorithm}

\begin{itemize}
\item \textsc{Reduction \alphafvs.R1.} Delete isolated vertices as they do not participate in any cycle.
\item \textsc{Reduction \alphafvs.R2.} If there is a vertex $v$ which has only
one neighbor $u$ in $G_i$, for some $i\in \onetoalpha$, then delete the edge $(v,u)$ from $E_i$.
\item \textsc{Reduction \alphafvs.R3.} If there is a vertex $v\in V(G)$ with exactly two neighbors $u,w$ (the total degree of $v$ is $2$),
delete edges $(v,u)$ and $(v,w)$ from $E_i$ and add an edge $(u,w)$ to $E_i$, where $i$ is the color of edges $(v,u)$ and $(v,w)$.
Note that after reduction \textsc{\alphafvs.R2} has been applied, both
edges $(v,u)$ and $(v,w)$ must be of the same color.
\item \textsc{Reduction \alphafvs.R4.} If for some $i$, $i \in \onetoalpha$, there is an edge of multiplicity larger than $2$ in $E_i$, reduce its multiplicity to $2$.
\item \textsc{Reduction \alphafvs.R5.} If there is a vertex $v$ with a self loop, then add $v$ to the solution set $X$, delete $v$ (and all edges incident on $v$) from the graph and decrease $k$ by $1$.
\end{itemize}

%\vspace*{2mm}
%\noindent
%bla
%\vspace*{2mm}

The safeness of reduction rule \textsc{\alphafvs.R4} follows from the
fact that edges of multiplicity greater than two do not influence the set of feasible solutions.
Safeness of reduction rule \textsc{\alphafvs.R5} follows from the fact
that any vertex with a loop must be present in every solution set $X$.
Note that all of the above reduction rules can be applied in polynomial time.
Moreover, after exhaustively applying all rules, the resulting graph $G$ satisfies the following properties:
\vspace*{1mm}
\newline(P1) $G$ contains no loops,
\newline(P2) Every edge in $G_i$, for $i \in \onetoalpha$ is of multiplicity at most two.
\newline(P3) Every vertex in $G$ has either degree zero or degree at least two in each $G_i$, for $i \in \onetoalpha$.
\newline(P4) The total degree of every vertex in $G$ is at least $3$.

%\begin{lemma}
%Reduction rule \textsc{\alphafvs.R1} is safe.
%\label{delete-0-vertex}
%\end{lemma}
%
%\begin{proof}
%For $G$ an $\alpha$-colored graph, we show that
%no optimal \alphafvssolshort\ contains a vertex $v$ of total degree $0$.
%Let $S$ be an optimal \alphafvssolshort~in $G$ and $v\in V(G)$ be a vertex of total degree $0$.
%Suppose $v \in S$ and let $S'=S\setminus \{v\}$.
%Observe that a vertex of total degree $0$ cannot participate in any cycle
%in any $G_i$, for $i \in \{1,2,\dots, \alpha\}$. Therefore, $S'$ intersects all cycles in $G_i$, for $i \in \{1,2,\dots,\alpha\}$.
%This implies that $S'$ is an \alphafvssolshort~with $\lvert S' \rvert <\lvert S \rvert$, contradicting the optimality of $S$.
%\end{proof}

\begin{lemma}
Reduction rule \textsc{\alphafvs.R2} is safe.
\label{delete-1-edge}
\end{lemma}

\begin{proof}
Let $G$ be an $\alpha$-colored graph and $v$ be a vertex whose only neighbor in $G_i$ is $u$, for some $i \in \onetoalpha$.
Consider the $\alpha$-colored graph $G'$ with vertex set $V(G)$ and edge sets $E_i(G')=E_i(G)\setminus \{(v,u)\}$
and $E_j(G')=E_j(G)$, for $j \in \onetoalphaNOi$.
We show that $G$ has an \alphafvssolshort~of size at most
$k$ if and only if $G'$ has an \alphafvssolshort~of size at most $k$.

In the forward direction, consider an \alphafvssolshort~$S$ in $G$ of size at most $k$.
Since $G'_j=G_j$, $S$ intersects all the cycles in $G'_j$, $j \in \onetoalphaNOi$.
Note that in $G_i$, there is no cycle containing the edge $(u,v)$ as $v$ is a degree-one vertex in $G_i$.
Hence, all the cycles in $G_i$ are also cycles in $G'_i$.
$S$ intersects all cycles in $G_i$ and, in particular, $S$ intersects all cycles in $G'_i$.
Therefore, $S$ is an \alphafvssolshort~in $G'$ of size at most $k$.

For the reverse direction, consider an \alphafvssolshort~$S$ in $G'$ of size at most $k$.
If $S$ is not an \alphafvssolshort~of $G$ then there is a cycle $C$ in some $G_t$, for $t \in \onetoalpha$.
Note that $C$ cannot be a cycle in $G_j$ as $G_j=G'_j$, for $j \in \onetoalphaNOi$.
Therefore $C$ must be a cycle in $G_i$. The cycle $C$ must contain the edge $(v,u)$, as this is the
only edge in $G_i$ which is not an edge in $G'_i$. But $v$ is a degree-one vertex in $G_i$, so it cannot
be part of any cycle in $G_i$, contradicting the existence of cycle $C$.
Thus $S$ is an \alphafvssolshort~of $G$ of size at most $k$.
\qed
\end{proof}

\begin{lemma}
Reduction rule \textsc{\alphafvs.R3} is safe.
\label{deg-2-rule}
\end{lemma}

\begin{proof}
Consider an $\alpha$-colored graph $G$. Let $v$ be a vertex in $V(G)$ such that
$v$ has total degree $2$ and let $u,w$ be the neighbors of $v$ in $G_i$, where $u \neq w$ and $i \in \onetoalpha$.
Consider the $\alpha$-colored graph $G'$ with vertex set $V(G)$ and edge sets
$E_i(G')=(E_i(G)\setminus \{(v,u), (v,w)\}) \cup \{(u,w)\}$ and $E_j(G')=E_j(G)$, for $j \in  \onetoalphaNOi$.
We show that $G$ has an \alphafvssolshort~of size at most $k$ if and only if $G'$ has an \alphafvssolshort~of size at most $k$.

In the forward direction, let $S$ be an \alphafvssolshort~in $G$ of size at most $k$.
Suppose $S$ is not an \alphafvssolshort~of $G'$. Then, there is a cycle $C$ in $G'_t$, for some $t \in \onetoalpha$.
Note that $C$ cannot be a cycle in $G'_j$ as $G'_j=G_j$, for $j \in \onetoalphaNOi$.
Therefore $C$ must be a cycle in $G'_i$. All the cycles $C'$ not containing the edge $(u,w)$ are also
cycles in $G_i$ and therefore $S$ must contain some vertex from $C'$. It follows that $C$ must contain
the edge $(u,w)$. Note that the edges $(E(C) \setminus \{(u,w)\}) \cup \{(v,u),(w,v)\}$ form a cycle in $G_i$.
Therefore $S$ must contain a vertex from $V(C)\cup \{v\}$. We consider the following cases:
\begin{itemize}
 \item Case 1: $v \notin S$. In this case $S$ must contains a vertex from $V(C)$. Hence, $S$ is an \alphafvssolshort~in $G'$.
 \item Case 2: $v \in S$. Let $S'=(S\setminus \{v\})\cup \{u\}$.
 Any cycle $C'$ containing $v$ in $G_i$ must contain $u$ and $w$ (since $d_{G_i}(v)=2$).
 But $S'$ intersects all such cycles $C'$, as $u \in S'$. Therefore $S'$ is an \alphafvssolshort~of $G'$ of size at most $k$.
\end{itemize}

In the reverse direction, consider an \alphafvssolshort~$S$ of $G'$. $S$ intersects all cycles in $G_j$, since $G_j=G'_j$, for $j \in\onetoalphaNOi$.
All cycles in $G_i$ not containing $v$ are also cycles in $G'_i$ and therefore $S$ intersects all such cycles.
A cycle $C$ in $G_i$ containing $v$ must contain $u$ and $w$ ($v$ is a degree-two vertex in $G_i$).
Note that $(E(C)\setminus \{(v,u)(v,w)\}) \cup \{(u,w)\}$ is a cycle in $G'_i$
and $S$, being an \alphafvssolshort~in $G'$, must contain a vertex from $V(C)\setminus \{v\}$.
Therefore $S \cap V(C)\neq \emptyset$, so $S$ intersects cycle $C$ in $G'_i$. Hence $S$ an \alphafvssolshort~in $G'$.
\qed
\end{proof}

%%%%%%%Algorithm description starts here%%%%%%%%%%%%%%%%%%%%%

\noindent \textbf{Algorithm:} We give an algorithm for the decision
version of the \disalphafvs\ problem, which only verifies whether a solution exists or not.
Such an algorithm can be easily modified to find an actual solution $X$.
We follow a branching strategy with a nontrivial measure function. Let $(G,W,k)$ be an instance of the problem, where $G$ is an $\alpha$-colored graph. If $G[W]$ is not an $\alpha$-forest then we can safely return that $(G,W,k)$ is a no-instance.
Hence, we assume that $G[W]$ is an $\alpha$-forest in what follows.
Whenever any of our reduction rules \alphafvs.R1 to \alphafvs.R5 apply, the algorithm exhaustively does so (in order). If at any point in our algorithm the parameter $k$ drops below zero, then the resulting instance is again a no-instance.

Recall that initially $F$ is an $\alpha$-forest, as $W$ is an \alphafvssolshort. We will consider each forest $F_i$, for $i \in \onetoalpha$, separately (where $F_i$ is the color $i$ graph of the $\alpha$-forest $F$).
For $i \in \onetoalpha$, we let $W_i=(W,E_i(G[W]))$ and $\eta_i$ be the number of components in $W_i$.
Some of the branching rules that we apply create special vertex-disjoint cycles.
We will maintain this set of special cycles in $\mathcal{C}_i$, for each $i$, and we
let $\mathfrak{C} = \{\mathcal{C}_i, \ldots, \mathcal{C}_{\alpha}\}$.
Initially, $\mathcal{C}_i=\emptyset$. Each cycle that we add to $\mathcal{C}_i$ will
be vertex disjoint from previously added cycles. Hence, if at any
point $\lvert \mathcal{C}_i \rvert >k$, for any $i$, then we can stop exploring the
corresponding branch. Moreover, whenever we ``guess'' that some vertex $v$ must belong to a solution,
we also traverse the family $\mathfrak{C}$ and remove any cycles containing $v$.
For the running time analysis of our algorithm we will consider the following measure:

$$\mu=\mu(G,W,k,\mathfrak{C})=\alpha k+(\sum_{i=1}^{\alpha}\eta_i)-(\sum_{i=1}^{\alpha}\lvert \mathcal{C}_i \rvert)$$

The input to our algorithm consists of a tuple $(G,W,k,\mathfrak{C})$. For clarity, we will denote a
reduced input by $(G,W,k, \mathfrak{C})$ (the one where reduction rules do not apply).
%, where $G$ is an $\alpha$-colored graph, $W$ is an \alphafvssol~in $G$ of size $k+1$, $k$ a positive integer and $\mathfrak{C}=\{\mathcal{C}_1,\mathcal{C}_2,\dots,\mathcal{C}_{\alpha}\}$, where initially each $\mathcal{C} \in \mathfrak{C}$ is an empty set.

We root each tree in $F_i$ at some arbitrary vertex. Assign an index $t$ to each vertex $v$ in the forest $F_i$, which is the distance of $v$ from the root of the tree it belongs to (the root is assigned index zero).
A vertex $v$ in $F_i$ is called \emph{\cordate} if one of the following holds:
\begin{itemize}
\item $v$ is a leaf (or degree-zero vertex) in $F_i$ with at least two color $i$ neighbors in $W_i$.
\item The subtree $T^i_v$ rooted at $v$ contains two vertices $u$ and $w$ which have at least one color $i$ neighbor in $W_i$ ($v$ can be equal to $u$ or $w$).
\end{itemize}

\begin{lemma} \label{cordate-path}
For $i \in \onetoalpha$, let $v_c$ be a \cordate~vertex of highest index in some tree of the forest $F_i$ and let ${\cal T}_{v_c}$ denote the subtree rooted at $v_c$. Furthermore, let $u_c$ be one of the vertices in ${\cal T}_{v_c}$ such that $u_c$ has a neighbor in $W_i$. Then, in the path $P=u_c,x_1,\ldots,x_t,v_c$ ($t$ could be equal to zero) between $u_c$ and $v_c$ the vertices $x_1,\ldots,x_t$ are degree-two vertices in $G_i$.
\end{lemma}
\begin{proof} Let $P=u_c,x_1,\ldots,x_t,v_c$ be the path from $v_c$ to $u_c$.
In $P$, if there is a vertex $x$ (other than $u_c$ and $v_c$) which has an edge of
color $i$ to a vertex in $W_i$, then $x$ is a \cordate~vertex of higher index, contradicting the choice of $v_c$.
Also, if there is a vertex $x$ in $P$ other than $v_c$ and $u_c$ of degree at least three in $F_i$, the subtree rooted at $x$ has at least two leaves, and all the leaves have a color-$i$ neighbor in $W_i$. Therefore, $x$ is a \cordate~vertex and has a higher index than $v_c$, contradicting the choice of $v_c$. It follows that $x_1,\ldots,x_t$ (if they exist) are degree-two vertices in $G_i$.
\qed
\end{proof}

We consider the following cases depending on whether there is a \cordate~vertex in $F_i$ or not.

%We consider the following cases depending on whether there is a \cordate vertex $v_c$ or not and $u_c$ and $w_c$ (if exists) having a color $i$ neighbor in the same or different component of $W_i$.

\begin{figure}
\centering
\includegraphics[scale=0.5]{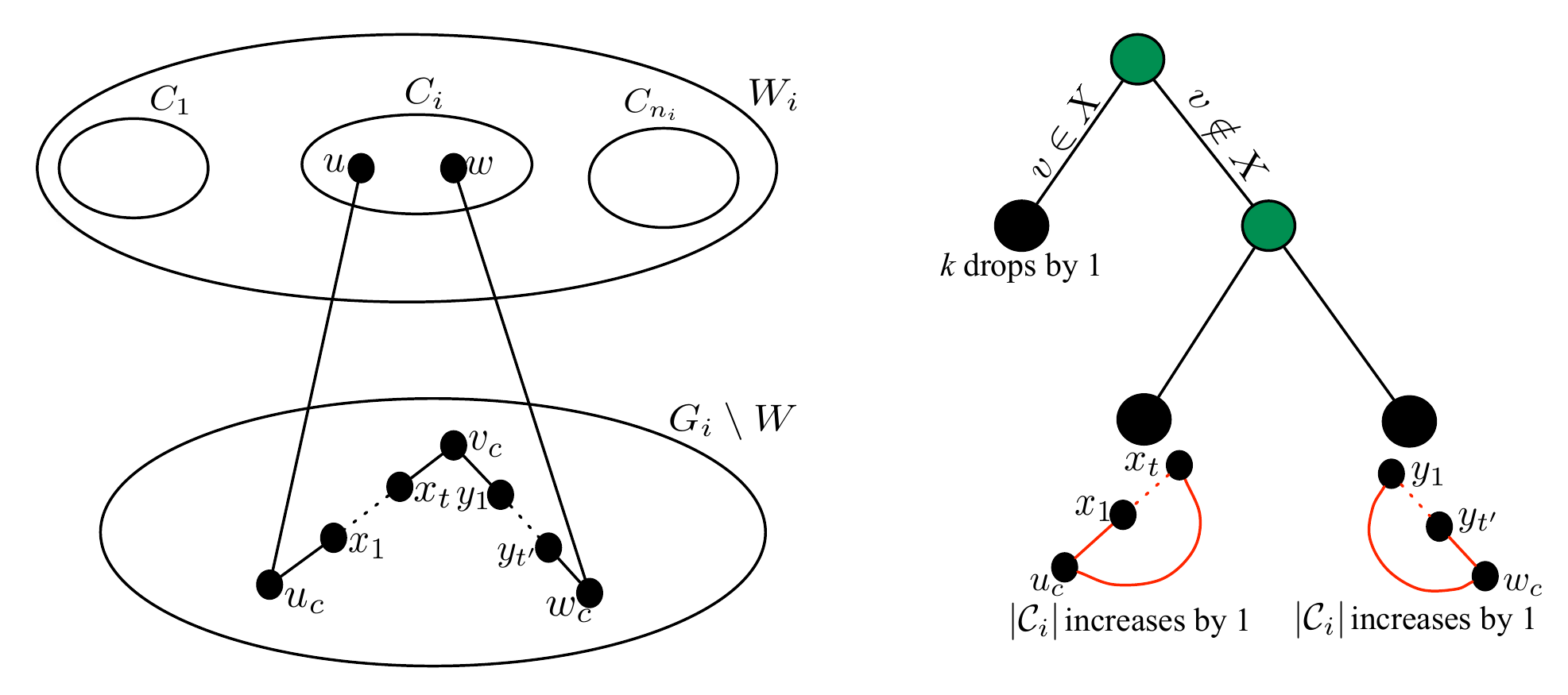}
\caption{Branching in Case 1.a}
\label{sim-fvs-case1a}
\end{figure}

\begin{itemize}
\item Case 1: There is a \cordate~vertex in $F_i$. Let $v_c$ be a \cordate~vertex
with the highest index in some tree in $F_i$ and let the two vertices
with neighbors in $W_i$ be $u_c$ and $w_c$ ($v_c$ can be equal to $u_c$ or $w_c$).
Let $P =u_c,x_1, x_2,$ $\cdots ,x_t,v_c$ and $P'$ $=$ $v_c$, $y_1$, $y_2$, $\cdots$ , $y_{t'}$, $w_c$
be the unique paths in $F_i$ from $u_c$ to $v_c$ and from $v_c$ to $w_c$, respectively.
Let $P_v=u_c,x_1,\cdots,x_t,$ $v_c,y_1,\cdots,y_{t'},w_c$ be the unique path in $F_i$
from $u_c$ to $w_c$. Consider the following sub-cases:

\vspace*{3mm}Case 1.a: $u_c$ and $w_c$ have a neighbor in the same component
of $W_i$. In this case one of the vertices from path $P_v$ must
be in the solution (Figure~\ref{sim-fvs-case1a}). We branch as follows:
    \begin{itemize}
	\item $v_c$ belongs to the solution. We delete $v_c$ from $G$ and decrease $k$ by $1$. In this branch $\mu$ decreases by $\alpha$.
    When $v_c$ does not belong to the solution, then at least one vertex
    from $u_c,x_1,x_2$, $\cdots ,x_t$ or $y_1,y_2,\cdots ,y_{t'},w_c$ must be in the solution.
    But note that these are vertices of degree at most two in $G_i$ by Lemma~\ref{cordate-path}.
    So with respect to color $i$, it does not matter which vertex is chosen in the solution. The only issue comes from some
    color $j$ cycle, where $j \neq i$, in which choosing a particular vertex
    from $u_c,x_1,\cdots ,x_t$ or $y_1,y_2,\cdots ,y_{t'},w_c$ would be more beneficial. We consider the following two cases.

    \item One of the vertices from $u_c,x_1,x_2,\cdots ,x_t$ is in the solution. In this
    case we add an edge $(u_c,x_t)$ (or $(u_c,u_c)$ when $u_c$ and $v_c$ are adjacent) to $G_i$ and delete the edge $(x_t,v_c)$ from $G_i$.
    This creates a cycle $C$ in $G_i\setminus W$, which is itself a component in $G_i \setminus W$.
    We remove the edges in $C$ from $G_i$ and add the cycle $C$ to $\mathcal{C}_i$. We will
    be handling these sets of cycles independently.
    In this case $\lvert \mathcal{C}_i \rvert$ increases by $1$, so the measure $\mu$ decreases by $1$.

    \item One of the vertices from $y_1,y_2,\cdots ,y_t,w_c$ is in the solution.
    In this case we add an edge $(y_1,w_c)$ to $G_i$ and delete the edge $(v_c,y_1)$
    from $G_i$. This creates a cycle $C$ in $G_i \setminus W$ as a component. We
    add $C$ to $\mathcal{C}_i$ and delete edges in $C$ from $G_i \setminus W$.
    In this branch $\lvert \mathcal{C}_i \rvert$ increases by $1$, so the measure $\mu$ decreases by $1$.\newline
    The resulting branching vector is \branchvector{\alpha,1,1}.
    \end{itemize}
\vspace*{3mm}Case 1.b: $u_c$ and $w_c$ do not have a neighbor
in the same component. We branch as follows (Figure~\ref{sim-fvs-case1b}):
    \begin{itemize}
    \item $v_c$ belongs to the solution. We delete $v_c$ from $G$ and
    decrease $k$ by $1$. In this branch $\mu$ decreases by $\alpha$.

    \item One of the vertices from $u_c,x_1,x_2,\cdots ,x_t$ is in the
    solution. In this case we add an edge $(u_c,x_t)$ to $G_i$ and delete
    the edge $(x_t,v_c)$ from $G_i$. This creates a cycle $C$ in $G_i\setminus W$ as a component.
    As in Case 1, we add $C$ to $\mathcal{C}_i$ and delete
    edges in $C$ from $G_i \setminus W$. $\lvert \mathcal{C}_i \rvert$ increases by $1$, so the measure $\mu$ decreases by $1$.

    \item One of the vertices from $y_1,y_2,\cdots ,y_t,w_c$ is in the solution.
    In this case we add an edge $(y_1,w_c)$ to $G_i$ and delete the edge
    $(v_c,y_1)$ from $G_i$. This creates a cycle $C$ in $G_i \setminus W$ as a component.
    We add $C$ to $\mathcal{C}_i$ and delete edges in $C$ from $G_i \setminus W$. In this
    branch $\lvert \mathcal{C}_i \rvert$ increases by $1$, so the measure $\mu$ decreases by $1$.

%    \item No vertex from paths $P$ or $P'$ is in the solution. In this case we
%    add the vertices in paths $P$ and $P'$ (apart from $v_c$, if $v_c$ is not a leaf) to $W$, the
%    resulting instance is $(G\setminus \{u_c,x_1,\dots,x_t,y_1,\dots,y_{t'},w_c\}, W \cup \{u_c,x_1,\dots,x_t,y_1,\dots,y_{t'},w_c\},k)$.
%    The number of components in $W_i$ decreases and we get a drop of $1$ in $\eta_i$, so $\mu$ decreases by $1$.\newline
%    The resulting branching vector is \branchvector{\alpha,1,1,1}.\newline
    \item No vertex from path $P_v$ is in the solution. In this case we
    add the vertices in $P_v$ to $W$, the resulting instance is $(G\setminus P_v, W \cup P_v,k)$.
    The number of components in $W_i$ decreases and we get a drop of $1$ in $\eta_i$, so $\mu$ decreases by $1$.
    Note that if $G[W \cup P_v]$ is not acyclic we can safely ignore this branch.\newline
    The resulting branching vector is \branchvector{\alpha,1,1,1}.\newline
    \end{itemize}

\item Case 2: There is no \cordate~vertex in $F_i$. Let $\mathcal{F}$ be a family
of sets containing a set $f_C = V(C)$ for each $C \in \cup_{i=1}^{\alpha} \mathcal{C}_i$
and let $\mathcal{U}=\cup_{i=1}^{\alpha} (\cup_{C \in \mathcal{C}_i}V(C))$.
Note that $\lvert \mathcal{F} \rvert \leq \alpha k$. We find a
subset $U \subseteq \mathcal{U}$ (if it exists) which hits all the
sets in $\mathcal{F}$, such that $\lvert U \rvert \leq k$.
%Return $U \cup X$ as the required solution.
\end{itemize}

\begin{figure}
\centering
\includegraphics[scale=0.5]{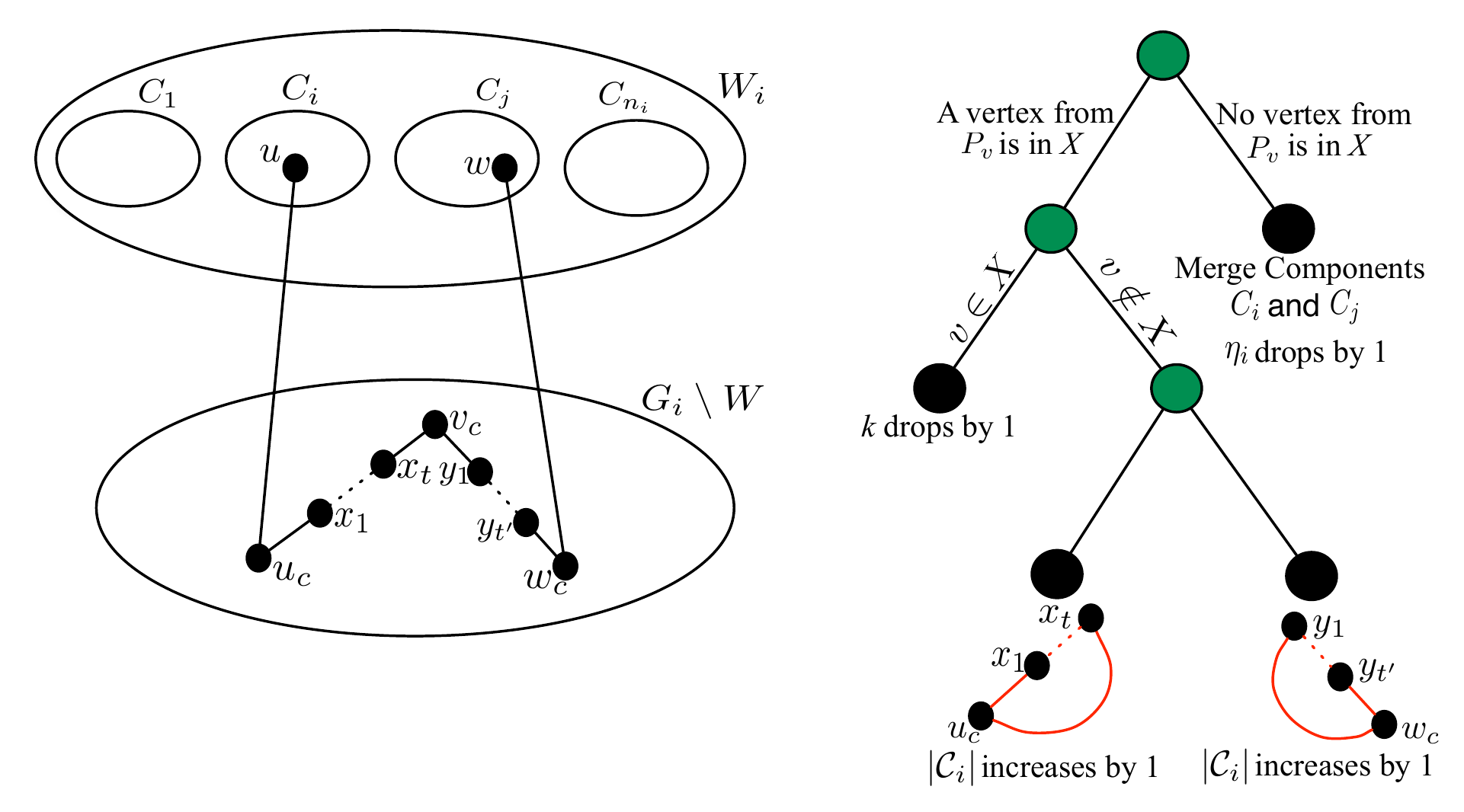}
\caption{Branching: Case 1.b}
\label{sim-fvs-case1b}
\end{figure}

Note that in Case 1, if the \cordate~vertex $v_c$ is a leaf, then $u_c=w_c=v_c$. Therefore,
from Case 1.a we are left with one branching rule. Similarly, we are left with the first and the last branching rules for Case 1.b.
If $v_c$ is not a leaf but $v_c$ is equal to $u_c$ or $w_c$, say $v_c=w_c$, then for both Case 1.a and Case 1.b we do
not have to consider the third branch.
Finally, when none of the reduction or branching rules apply, we solve the problem by invoking an algorithm for the
{\sc Hitting Set} problem as a subroutine.

\begin{lemma}
The presented algorithm for \disalphafvs\ is correct.
\end{lemma}

\begin{proof}
Consider an input $(G,W,k,\mathfrak{C})$ to the algorithm for \disalphafvs, where $G$
is an $\alpha$-colored graph, $W$ is an \alphafvssolshort\ of size $k+1$, and $k$ is a positive integer and $\mathfrak{C}=\{\mathcal{C}_1,\mathcal{C}_2,\dots,\mathcal{C}_1\}$.
Let $\mu =\mu(G,W,k,\mathfrak{C})$ be the measure as defined earlier.
We prove the correctness of the algorithm by induction on the measure $\mu$.
The base case occurs when one of the following holds:
\begin{itemize}
\item $k<0$,
\item for some $i \in \onetoalpha$, $\lvert \mathcal{C}_i \rvert > k$, or
\item $\mu \leq 0$.
\end{itemize}
If $k < 0$, then we can safely conclude that $G$ is a no-instance.
If for some $i \in \onetoalpha$ we have $\lvert \mathcal{C}_i \rvert >k$, then we need
to pick at least one vertex from each of the vertex-disjoint cycles in $\mathcal{C}_i$ and
there are at least $k+1$ of them. Our algorithm correctly concludes that the graph is also a no-instance in such cases.
If $\mu=\alpha k+(\sum_{i=1}^{\alpha}\eta_i)-(\sum_{i=1}^{\alpha}\lvert \mathcal{C}_i \rvert) \leq 0$
then $\alpha k \leq \sum_{i=1}^{\alpha}\lvert \mathcal{C}_i \rvert$.
But for each $i \in \onetoalpha$, we have $\lvert \mathcal{C}_i \rvert \leq k$.
Therefore $\alpha k \leq \sum_{i=1}^{\alpha}\lvert \mathcal{C}_i \rvert \leq \alpha k$,
$\sum_{i=1}^{\alpha}\lvert \mathcal{C}_i \rvert= \alpha k$,
and $|\mathcal{C}_i| = k$, for all $i \in \onetoalpha$.
This implies that for each $i \in \onetoalpha$, $G_i[F_i \cup W_i]$ must be acyclic.
Assume otherwise. Then, for some $i \in \onetoalpha$, $G_i[F_i \cup W_i]$ contains a cycle which is vertex disjoint from
the $k$ cycles in $\mathcal{C}_i$. Therefore, at least $k + 1$ vertices are needed to intersect these cycles and
we again have a no-instance. Recall that when a new vertex $v$ is added to the solution set we delete all
those cycles in $\cup_{i=1}^{\alpha} \mathcal{C}_i$ which contain $v$.

We are now left with cycles in $\cup_{i=1}^{\alpha} \mathcal{C}_i$.
Intersecting a cycle $C \in \cup_{i=1}^{\alpha} \mathcal{C}_i$ is equivalent to hitting the set $V(C)$.
Hence, we construct a family $\mathcal{F}$ consisting of a set $f_C = V(C)$ for each $C \in \cup_{i=1}^{\alpha} \mathcal{C}_i$ and
we let $\mathcal{U}=\cup_{i=1}^{\alpha} (\cup_{C \in \mathcal{C}_i}V(C))$.
Note that $\lvert \mathcal{F} \rvert \leq \alpha k$. If we can find a subset $U \subseteq \mathcal{U}$ which
hits all the sets in $\mathcal{F}$, such that $\lvert U \rvert \leq k$, then $U$ is the required solution.
Otherwise, we have a no-instance.
It is known that the {\sc Hitting Set} problem parameterized by the
size of the family $\mathcal{F}$ is fixed-parameter tractable and
can be solved in $\OO^\star(2^{|\mathcal{F}|})$ time~\cite{saurabh-book}. In particular, we can find an
optimum hitting set $U \subseteq \mathcal{U}$, hitting all the sets in $\mathcal{F}$.
Therefore, we have a subset of vertices that intersects all the cycles in $\mathcal{C}_i$, for $i \in \onetoalpha$.

Putting it all together, at a base case, our algorithm correctly decides whether $(G,W,k,\mathfrak{C})$ is a yes-instance or not.
For the induction hypothesis, assume that the algorithm correctly decides an instance for $\mu \leq t$.
Now consider the case $\mu=t+1$. If some reduction rule applies then we create an equivalent instance (since all reduction rules are safe).
Therefore, either we get an equivalent instance with the same measure or we get an equivalent instance with $\mu \leq t$
(the case when \alphafvs.R5 is applied).
In the latter case, by the induction hypothesis, our algorithm correctly decides the instance where $\mu \leq t$.
In the former case, we apply one of the branching rules. Each branching rule is exhaustive and
covers all possible cases. In addition, the measure decreases at each branch by at least one.
Therefore, by the induction hypothesis, the algorithm correctly decides whether the input is a yes-instance or not.
\qed
\end{proof}

\begin{lemma}
\disalphafvs~is solvable in time $\OO^\star(22^{\alpha k})$.
\end{lemma}

\begin{proof}
All of the reduction rules \alphafvs.R1 to \alphafvs.R5 can be applied in time polynomial in the input size.
Also, at each branch we spend a polynomial amount of time.
For each of the recursive calls at a branch, the measure $\mu$ decreases at least by $1$.
When $\mu \leq 0$, then we are able to solve the remaining instance in time $\OO(2^{\alpha k})$ or
correctly conclude that the corresponding branch cannot lead to a solution.
At the start of the algorithm $\mu \leq 2\alpha k$.
Therefore, the height of the search tree is bounded by $2\alpha k$.
The worst-case branching vector for the algorithm
is \branchvector{\alpha,1,1,1}. The recurrence relation for the worst case branching vector is:
$T(\mu)\leq T(\mu - \alpha)+3T(\mu -1) \leq T(\mu - 2)+3T(\mu -1)$, since $\alpha \geq 2$.  % otherwise it is the feedback vertex set problem.
The running time corresponding to the above recurrence relation is $3.303^{2\alpha k}$.
At each branch we spend a polynomial amount of time but
we might require $\OO(2^{\alpha k})$ time. for solving the base case.
Therefore, the running time of the algorithm is $\OO^{\star}(2^{\alpha k}\cdot 3.303^{2 \alpha k})=\OO^{\star}(22^{\alpha k})$.
%We can optimize the running time if we take the measure to be the following:
%
%$$\mu=\alpha k+(\sum_{i=1}^{\alpha}(\eta_i/ \alpha))-(\sum_{i=1}^{\alpha}\lvert \mathcal{C}_i \rvert)$$
%
%The worst case branching vector corresponding to the above measure
%is $(\alpha, 1, 1, 1/ \alpha)$. The corresponding recurrence relation for the branching vector is:
%$T(\mu)\leq T(\mu - \alpha)-2T(\mu -1) + T(\mu -1/\alpha)$.
\qed
\end{proof}

\begin{theorem}
\alphafvsfull~is solvable in time $\OO^\star(23^{\alpha k})$.
\end{theorem}

%In the following subsection we will consider the
%case $\alpha=2$ and show that the base case can be solved in polynomial time.
\subsection{Faster algorithm for \twofvsfull}

We improve the running time of the \FPT\ algorithm for \alphafvs\ when $\alpha=2$.
Given two sets of disjoint cycles $\mathcal{C}_1$ and $\mathcal{C}_2$ and a
set $V=\cup_{C \in \mathcal{C}_1 \cup \mathcal{C}_2} V(C)$,
we want to find a subset $H \subseteq V$ such that $H$ contains at least
one vertex from $V(C)$, for each $C \in \mathcal{C}_1 \cup \mathcal{C}_2$.
We construct a bipartite graph $G_M$ as follows.
We set $V(G_M)= \{c^1_x | C_x \in \mathcal{C}_1\} \cup \{c^2_y | C_y \in \mathcal{C}_2\}$.
In other words, we create one vertex for each cycle in $\mathcal{C}_1 \cup \mathcal{C}_2$.
We add an edge between $c^1_x$ and $c^2_y$ if and only if
$V(C_x) \cap V(C_y) \neq \emptyset$.
Note that for $i\in \{1,2\}$ and $C,C' \in \mathcal{C}_i$, $V(C)\cap V(C')= \emptyset$.
In Lemma~\ref{2-hitting}, we show that finding a matching $M$ in $G_M$, such
that $\lvert M \rvert + \lvert V(G_M) \setminus V(M) \rvert \leq k$, corresponds to
finding a set $H$ of size at most $k$, such that $H$ contains at
least one vertex from each cycle $C \in \mathcal{C}_1 \cup \mathcal{C}_2$.

%We construct a graph $G_M$ such that finding a maximum matching in $G_M$ corresponds to finding a minimum hitting set in $(\mathcal{F},\mathcal{U})$, where $\mathcal{F}=\{V(C)\lvert C \in \mathcal{C}_1 \cup \mathcal{C}_2\}$ and $\mathcal{U}=\cup_{C \in \mathcal{C}_1 \cup \mathcal{C}_2} V(C)$. We construct $G_M$ as the following. $V(G_M)=\mathcal{C}_1 \cup \mathcal{C}_2$. We add an edge between $C_1 \in \mathcal{C}_1$ and $C_2 \in \mathcal{C}_2$ if and only if $V(C_1) \cap V(C_2) \neq \emptyset$. Note that for $i\in \{1,2\}$ and $C,C' \in \mathcal{C}_i$, $V(C)\cap V(C')= \emptyset$. show that

\begin{lemma} \label{2-hitting}
For $i \in \{1,2\}$, let $\mathcal{C}_i$ be a set of vertex-disjoint cycles, i.e.
for each $C,C' \in \mathcal{C}_i$, $C \neq C'$ implies $V(C)\cap V(C')=\emptyset$.
Let $\mathcal{F}=\{V(C)\lvert C \in \mathcal{C}_1 \cup \mathcal{C}_2\}$
and $\mathcal{U}=\cup_{C \in \mathcal{C}_1 \cup \mathcal{C}_2} V(C)$.
There exists a vertex subset $H \subseteq \cup_{C \in \mathcal{C}_1 \cup \mathcal{C}_2} V(C)$ of size $k$ such that $H\cap V(C) \neq \emptyset$,
for each $C \in \mathcal{C}_1 \cup \mathcal{C}_2$, if and only if $G_M$ has a matching $M$, such
that $\lvert M \rvert + \lvert V(G_M) \setminus V(M) \rvert \leq k$.
\end{lemma}

\begin{proof}
For the forward direction, consider a minimal vertex subset $H \subseteq V(C_1) \cup V(C_2)$
of size at most $k$ such that for each $C \in \mathcal{C}_1 \cup \mathcal{C}_2$, $H\cap V(C) \neq \emptyset$.
Note that a vertex $h \in H$ can be present in at most one cycle
from $\mathcal{C}_i$, for $i \in \{1,2\}$, since $\mathcal{C}_i$ is a set
of vertex-disjoint cycles. Therefore, $h$ can be present in at most $2$
cycles from $\mathcal{C}_1 \cup \mathcal{C}_2$. If $h$ is present in $2$
cycles, say $C_x \in \mathcal{C}_1$ and $C_y \in \mathcal{C}_2$, then in
$G_M$ we must have an edge between $c^1_x$ and $c^2_y$ (since $h$ belongs to
both $C_x$ and $C_y$). We include the edge $(c^1_x,c^2_y)$ in the matching $M$.
If $h$ belongs to only one cycle, say $C^i_z \in \mathcal{C}_1\cup \mathcal{C}_2$, then we include
vertex $c^i_z$ in a set $I$. Note that $(V(G_M)\setminus V(M)) \subseteq I$. For each $h \in H$, we either add a
matching edge or add a vertex to $I$. Therefore
$\lvert M \rvert + \lvert V(G_M)\setminus V(M) \rvert \leq \lvert M \rvert + \lvert I \rvert \leq k$.

In the reverse direction, consider a matching
$M$ such that $\lvert M \rvert + \lvert V(G_M)\setminus V(M) \rvert \leq k$. We
construct a set $H$ of size at most $k$ containing a vertex from
each cycle in $\mathcal{C}_1 \cup \mathcal{C}_2$. For each edge $(c^1_x, c^2_y)$ in the
matching, where $C_x \in \mathcal{C}_1$ and $C_y \in \mathcal{C}_2$, there
is a vertex $h$ that belongs to both $V(C_x)$ and $V(C_y)$. Include $h$ in $H$.
For each $c^i_z \in V(G_M) \setminus V(M)$, add an arbitrary
vertex $v \in V(C_z)$ to $H$. Note that $\lvert H \rvert \leq k$, since for each
matching edge and each unmatched vertex we added one vertex to $H$. Moreover, for each
cycle $C \in \mathcal{C}_1 \cup \mathcal{C}_2$, its corresponding vertex in $G_M$ is either part of the matching
or is an unmatched vertex; in both cases there is a vertex in $H$
that belongs to $C$. Therefore, $H$ is a subset of size at most $k$ which
contains at least one vertex from each cycle in $\mathcal{C}_1 \cup \mathcal{C}_2$.
\qed
\end{proof}

Note that a matching $M$ in $G_M$ minimizing
$\lvert M \rvert + \lvert V(G_M) \setminus V(M) \rvert$ is one of maximum size.
Therefore, at the base case for \twofvs~we compute a maximum matching of the
corresponding graph $G_M$, which is a polynomial-time solvable problem, and return an optimal solution for
intersecting all cycles in $\mathcal{C}_1 \cup \mathcal{C}_2$.
Moreover, if we set $\mu = 2k+(\eta_1/\alpha+\eta_2/\alpha)-(\lvert \mathcal{C}_1 \rvert + \lvert \mathcal{C}_2 \rvert)$, then the
worst case branching vector is \branchvector{2,1,1,1/2}.
Corresponding to this worst case branching vector, the running time of the algorithm is $\OO^{\star}(81^{k})$.

\begin{theorem}
\twofvsfull~is solvable in time $\OO^{\star}(81^{k})$.
\end{theorem}

\section{Polynomial kernel for \alphafvsfull}\label{sec-kernel}
In this section we give a kernel with $\OO(\alpha k^{3(\alpha + 1)})$ vertices
for \alphafvs. Let $(G,k)$ be an instance of \alphafvs, where
$G$ is an $\alpha$-colored graph and $k$ is a positive integer. We assume that
reduction rules \textsc{\alphafvs.R1} to \textsc{\alphafvs.R5} have been exhaustively applied.
The kernelization algorithm then proceeds in two stages.
In stage one, we bound the maximum degree of $G$. In the second stage, we present new
reduction rules to deal with degree-two vertices and conclude a bound on the total number
of vertices.

To bound the total degree of each vertex $v \in V(G)$, we bound
the degree of $v$ in $G_i$, for $i \in \{1,2,\dots,\alpha\}$.
To do so, we need the Expansion Lemma~\cite{saurabh-book} as well as the
$2$-approximation algorithm for the classical {\sc{Feedback Vertex Set}} problem~\cite{2-approx-fvs-bafna}.

A {\em $q$-star}, $q \geq 1$, is a graph with $q + 1$ vertices, one vertex of degree $q$
and all other vertices of degree $1$. Let $G$ be a bipartite graph
with vertex bipartition $(A,B)$. A set of edges $M \subseteq E(G)$ is called a {\em $q$-expansion}
of $A$ into $B$ if (i) every vertex of $A$ is incident with exactly $q$ edges of M and (ii)
$M$ saturates exactly $q|A|$ vertices in $B$.

\begin{lemma}[Expansion Lemma~\cite{saurabh-book}]\label{qexpansion}
Let $q$ be a positive integer and $G$
be a bipartite graph with vertex bipartition $(A,B)$ such that
$|B| \geq q|A|$ and there are no isolated vertices in $B$.
Then, there exist nonempty vertex sets $X \subseteq A$ and $Y \subseteq B$ such that:
\begin{itemize}
\item (1) $X$ has a $q$-expansion into $Y$ and
\item (2) no vertex in $Y$ has a neighbour outside $X$, i.e. $N(Y) \subseteq X$.
\end{itemize}
Furthermore, the sets $X$ and $Y$ can be found in time polynomial in the size of $G$.
\end{lemma}

\subsection{Bounding the degree of vertices in $G_i$}
We now describe the reduction rules that allow us to bound the maximum
degree of a vertex $v\in V(G)$.

\begin{lemma}[Lemma 6.8~\cite{evencycle}]
Let $G$ be an undirected multi-graph and $x$ be a vertex of $G$ without a self loop.
Then in polynomial time we can either decide that $(G,k)$ is a no-instance
of \textsc{Feedback Vertex Set} or check whether there is an $x$-flower of
order $k+1$, or find a set of vertices $Z\subseteq V(G)\setminus \{x\}$ of
size at most $3k$ intersecting every cycle in $G$, i.e. $Z$ is a feedback vertex set of $G$.
\label{set-H-exp-lemma-gen}
\end{lemma}

The next proposition easily follows from Lemma~\ref{set-H-exp-lemma-gen}.

\begin{proposition}
Let $G$ be an undirected $\alpha$-colored multi-graph
and $x$ be a vertex without a self loop in $G_i$, for $i \in \{1,2,\dots, \alpha\}$.
Then in polynomial time we can either decide that $(G,k)$ is a no-instance
of \alphafvsfull~or check whether there is an $x$-flower of order $k+1$ in $G_i$, or find a set
of vertices $Z\subseteq V(G)\setminus \{x\}$ of size at most $3k$ intersecting every cycle in $G_i$.
\label{set-H-exp-lemma}
\end{proposition}

After applying reduction rules \textsc{\alphafvs.R1} to \textsc{\alphafvs.R5} exhaustively, we know that the
degree of a vertex in each $G_i$ is either $0$ or at least $2$ and no vertex has a self loop.
Now consider a vertex $v$ whose degree in $G_i$ is more than $3k(k+4)$.
By Proposition~\ref{set-H-exp-lemma}, we know that one of three cases must apply:
\begin{itemize}
\item (1) $(G,k)$ is a no-instance of \alphafvs,
\item (2) we can find (in polynomial time) a $v$-flower of order $k+1$ in $G_i$, or
\item (3) we can find (in polynomial time) a set $H_v \subseteq V(G_i)$ of size at
most $3k$ such that $v \notin H_v$ and $G_i \setminus H_v$ is a forest.
\end{itemize}
The following reduction rule allows us to deal with case (2).
The safeness of the rule follows from the fact that if $v$ in not included in the solution
then we need to have at least $k+1$ vertices in the solution.

\vspace*{2mm}
\noindent \textsc{Reduction \alphafvs.R6.} For $i \in \{1,2,\dots, \alpha\}$, if $G_i$ has a vertex $v$ such that there is a $v$-flower of order at least $k+1$  in $G_i$, then include $v$ in the solution $X$ and decrease $k$ by $1$. The resulting instance is $(G\setminus \{v\},k-1)$.
\vspace*{2mm}

When in case (3), we bound the degree of $v$ as follows.
Consider the graph $G'_i=G_i \setminus (H_v \cup \{v\} \cup V^i_0)$, where $V^i_0$ is
the set of degree $0$ vertices in $G_i$. Let $\mathcal{D}$ be the set of components in
the graph $G'_i$ which have a vertex adjacent to $v$ . Note that each $D\in \mathcal{D}$
is a tree and $v$ cannot have two neighbors in $D$, since $H_v$ is a feedback vertex set in $G_i$.
We will now argue that each component $D\in \mathcal{D}$ has a
vertex $u$ such that $u$ is adjacent to a vertex in $H_v$. Suppose for a contradiction that there is a
component $D\in \mathcal{D}$ such that $D$ has no vertex which is adjacent to a vertex in $H_v$.
$D \cup \{v\}$ is a tree with at least $2$ vertices, so $D$ has a vertex $w$, such that $w$ is a degree-one vertex in $G_i$,
contradicting the fact that each vertex in $G_i$ is either of degree zero or of degree at least two.

After exhaustive application of \textsc{\alphafvs.R4}, every pair of vertices in $G_i$ can
have at most two edges between them. In particular, there can be at most two edges
between $h \in H_v$ and $v$. If the degree of $v$ in $G_i$ is more than $3k(k+4)$, then the
number of components $\lvert \mathcal{D} \rvert$, in $G'_i$ is more than $3k(k+2)$, since $\lvert H_v \rvert \leq 3k$.

Consider the bipartite graph $\mathcal{B}$, with bipartition $(H_v,Q)$,
where $Q$ has a vertex $q_D$ corresponding to each component $D\in \mathcal{D}$.
We add an edge between $h \in H_v$ and $q_D \in Q$ to
$E(\mathcal{B})$ if and only if $D$ has a vertex $d$ which is adjacent to $h$ in $G_i$.

\vspace*{2mm}
\noindent
\textsc{Reduction \alphafvs.R7.} Let $v$ be a vertex of degree
at least $3k(k+4)$ in $G_i$, for $i \in \{1,2,\dots, \alpha \}$, and let $H_v$ be a
feedback vertex set in $G_i$ not containing $v$ and of size at most $3k$.
\begin{itemize}
\item Let $Q' \subseteq Q$ and $H \subseteq H_v$ be the sets of vertices
obtained after applying Lemma~\ref{qexpansion} with $q=k+2$, $A=H_v$, and $B=Q$, such
that $H$ has a $(k+2)$-expansion into $Q'$ in $\mathcal{B}$;
\item Delete all the edges $(d,v)$ in $G_i$, where $d \in V(D)$ and $q_D \in Q'$;
\item Add double edges between $v$ and $h$ in $G_i$, for all $h \in H$ (unless such edges already exist).
\end{itemize}
\vspace*{2mm}

By Lemma~\ref{qexpansion} and Proposition~\ref{set-H-exp-lemma}, \textsc{\alphafvs.R7} can
be applied in time polynomial in the input size.
%In Lemma~\ref{safe-r7} we show that reduction \textsc{\alphafvs.R7} is safe and the resulting instance is an equivalent instance.

\begin{lemma}
Reduction rule \textsc{\alphafvs.R7} is safe.
\label{safe-r7}
\end{lemma}

\begin{proof}
Let $G$ be an $\alpha$-colored graph where reductions \textsc{\alphafvs.R1} to \textsc{\alphafvs.R6} do not apply.
Let $v$ be a vertex of degree more than $3k(k+4)$ in $G_i$, for $i \in \{1,2,\dots, \alpha \}$.
Let $H \subseteq H_v$, $Q' \subseteq Q$ be the sets defined above and let
$G'$ be the instance obtained after a single application of reduction rule \textsc{\alphafvs.R7}.
We show that $G$ has an \alphafvssolshort~of size at most $k$ if and only if $G'$ has an \alphafvssolshort~of size at most $k$.
We need the following claim.

\begin{claim}
Any $k$-sized \alphafvssolshort\ $S$ of $G$ or $G'$ either contains $v$ or contains all the vertices in $H$.
\end{claim}

\begin{proof}
Since there exists a cycle (double edge) between $v$ and every vertex $h \in H$ in $G'_i$,
it easily follows that either $v$ or all vertices in $H$ must be in any solution for $G'$.

Consider the case of $G$. We assume $v \notin S$ and
there is a vertex $h \in H$ such that $h \notin S$. Note that
$H$ has a $(k+2)$-expansion into $Q'$ in $\mathcal{B}$, therefore $h$ is the center of
a $(k+2)$-star in $\mathcal{B}[H \cup Q']$.
Let $Q_h$ be the set of neighbors of $h$ in $\mathcal{B}[H \cup Q']$ ($\lvert Q_h \rvert \geq k+2$).
For each $q_D, q_{D'} \in Q_h$, their corresponding components $D,D'\in \mathcal{D}$
form a cycle with $v$ and $h$.   %, pairwise intersecting at only $v$ and $h$.
If both $h$ and $v$ are not in $S$, then we need to pick at least $k+1$
vertices to intersect the cycles formed by $D$, $D'$, $h$, and $v$, for each $q_D, q_{D'} \in Q'$.
Therefore, $H \subseteq S$, as needed.
\qed
\end{proof}

In the forward direction, consider an \alphafvssolshort\ $S$ of size at most $k$ in $G$.
For $j \in \{1,2,\dots,\alpha \} \setminus \{i\}$, $G'_j = G_j$ and therefore $S$ intersects all the cycles in $G'_j$.
By the previous claim, we can assume that either $v \in S$ or $H \subseteq S$.
In both cases, $S$ intersects all the new cycles created in $G'_i$ by adding double edges between $v$ and $h \in H$.
Moreover, apart from the double edges between $v$ and $h \in H$, all the cycles in
$G'_i$ are also cycles in $G_i$, therefore $S$ intersects all those
cycles in $G'_i$. It follows that $S$ is an \alphafvssolshort\ in $G'$.

In the reverse direction, consider an \alphafvssolshort~$S$ in $G'$ of size at most $k$.
Note that for $j \in \{1,2,\dots,\alpha \} \setminus \{i\}$, $G'_j=G_j$. Therefore $S$ intersects all the cycles in $G_j$.
By the previous claim, at least one of the following must hold: (1) $v \in S$ or (2) $H \subseteq S$.

Suppose that (1) $v \in S$. Since $G'_i\setminus \{v\}=G_i \setminus \{v\}$, $S \setminus \{v\}$
intersects all the cycles in $G'_i \setminus \{v\}$ and $G_i\setminus \{v\}$. Therefore
$S$ intersects all the cycles in $G_i$ and $S$ is an \alphafvssolshort~in $G$.
In case (2), i.e. when $v \not\in S$ but $H \subseteq S$, any cycle in $G$
which does not intersect with $S$ is also a cycle in $G'$ (since such a cycle
does not intersect with $H$ and the only deleted
edges from $G'$ belong to cycles passing through $H$). In other words,
$S \setminus H$ intersects all cycles in both $G'_i \setminus H$ and $G_i\setminus H$
and, consequently, $S$ is an \alphafvssolshort\ in $G$.
\qed
\end{proof}

After exhaustively applying all reductions \textsc{\alphafvs.R1} to \textsc{\alphafvs.R7}, the
degree of a vertex $v\in V(G_i)$ is at most $3k(k+4)-1$ in $G_i$, for $i \in \{1,2,\dots, \alpha\}$.

\subsection{Bounding the number of vertices in $G$}
Having bounded the maximum total degree of a vertex in $G$, we now focus on bounding the number
of vertices in the entire graph. To do so, we first compute an
approximate solution for the \alphafvs\ instance using
the polynomial-time $2$-approximation algorithm of Bafna et al.~\cite{2-approx-fvs-bafna} for
the {\sc{Feedback Vertex Set}} problem in undirected graphs.
In particular, we compute a $2$-approximate solution $S_i$ in $G_i$, for $i \in \{1,2, \dots, \alpha\}$.
We let $S=\cup_{i=1}^{\alpha} S_i$. Note that $S$ is an \alphafvssolshort~in $G$ and has size
at most $2\alpha \lvert S_{OPT}\rvert$, where $\lvert S_{OPT} \rvert$ is an optimal \alphafvssolshort~in $G$.
Let $F_i=G_i \setminus S_i$. Let $T^i_{\leq 1}$, $T^i_{2}$, and $T^i_{\geq 3}$, be the sets of
vertices in $F_i$ having degree at most one in $F_i$, degree exactly two in $F_i$,
and degree greater than two in $F_i$, respectively.

Later, we shall prove that bounding the maximum degree in $G$ is sufficient for bounding
the sizes of $T^i_{\leq 1}$ and $T^i_{\leq 1}$, for all $i \in \{1,2, \dots, \alpha\}$.
We now focus on bounding the size of $T^i_{2}$ which, for each $i \in \{1,2, \dots, \alpha\}$, corresponds
to a set of degree-two paths. In other words, for a fixed $i$, the graph induced by the vertices in $T^i_{2}$
is a set of vertex-disjoint paths. We say a set of distinct vertices $\{v_1, \ldots, v_\ell\}$
in $T^i_{2}$ forms a {\em maximal degree-two path} if $(v_j,v_{j+1})$ is an edge, for all $1 \leq j \leq \ell$,
and all vertices $\{v_1, \ldots, v_\ell\}$ have degree exactly two in $G_i$.

We enumerate all the maximal degree-two paths in $G_i \setminus S_i$, for $i \in \{1,2,\dots, \alpha \}$.
Let this set of paths in $G_i\setminus S_i$ be $\mathcal{P}_i=\{P^i_1,P^i_2, \dots, P^i_{n_i}\}$, where
$n_i$ is the number of maximal degree-two paths in $G_i \setminus S_i$. We introduce a special symbol $\phi$ and add
$\phi$ to each set $\mathcal{P}_i$, for $i \in \{1,2,\dots, \alpha\}$.
The special symbol will be used later to indicate that no path is chosen from the set $\mathcal{P}_i$.

Let $\mathfrak{S}= \mathcal{P}_1 \times \mathcal{P}_2 \times \dots \times \mathcal{P}_{\alpha}$ be the
set of all tuples of maximal degree-two paths of different colors.
For $\tau \in \mathfrak{S}$, $j \in \{1,2,\dots, \alpha \}$, $j(\tau)$ denotes the element from
the set $\mathcal{P}_j$ in the tuple $\tau$, i.e. for $\tau =(Q_1,\phi,\dots, Q_j,\dots, Q_{\alpha})$, $j(\tau)=Q_j$ (for example $2(\tau)=\phi$).

For a maximal degree-two path $P^i_j \in \mathcal{P}_i$ and $\tau \in \mathfrak{S}$, we define 
\emph{Intercept}$(P^i_j,\tau)$ to be the set of vertices in path $P^i_j$ which are present 
in all the paths in the tuple (of course a $\phi$ entry does not contribute to this set).
Formally, \emph{Intercept}$(P^i_j,\tau) = \emptyset$ if $P^i_j \not\in \tau$ otherwise
\emph{Intercept}$(P^i_j,\tau)=\{v \in V(P^i_j) \lvert$ for all $1 \leq t \leq \alpha$, if $t(\tau) \neq \phi$ then $v \in V(t(\tau))\}$.

We define the notion of \emph{unravelling} a path $P^i_j \in \mathcal{P}_i$ from all
other paths of different colors in $\tau \in \mathfrak{S}$ at a vertex
$u\in$ \emph{Intercept}$(P^i_j,\tau)$ by creating a separate copy of $u$ for each path.
Formally, for a path $P^i_j \in \mathcal{P}_i$, $\tau\in \mathfrak{S}$, and a vertex
$u \in$\emph{Intercept}$(P^i_j,\tau)$, the \emph{Unravel}$(P^i_j,\tau,u)$ operation does the following.
For each $t \in \{1,2,\dots, \alpha \}$ let $x_t$ and $y_t$ be the unique neighbors
of $u$ on path $t(\tau)$. Create a vertex $u_{t(\tau)}$ for each path $t(\tau)$, for $1 \leq t \leq \alpha$, delete
the edges $(x_t,u)$ and $(u,y_t)$ from $G_t$ and add the edges $(x_t,u_{t(\tau)})$ and $(u_{t(\tau)},y_t)$ in $G_t$.
Figure~\ref{fig-unravel} illustrates the unravel operation for two paths of different colors.

\begin{figure}
\centering
\includegraphics[scale=0.4]{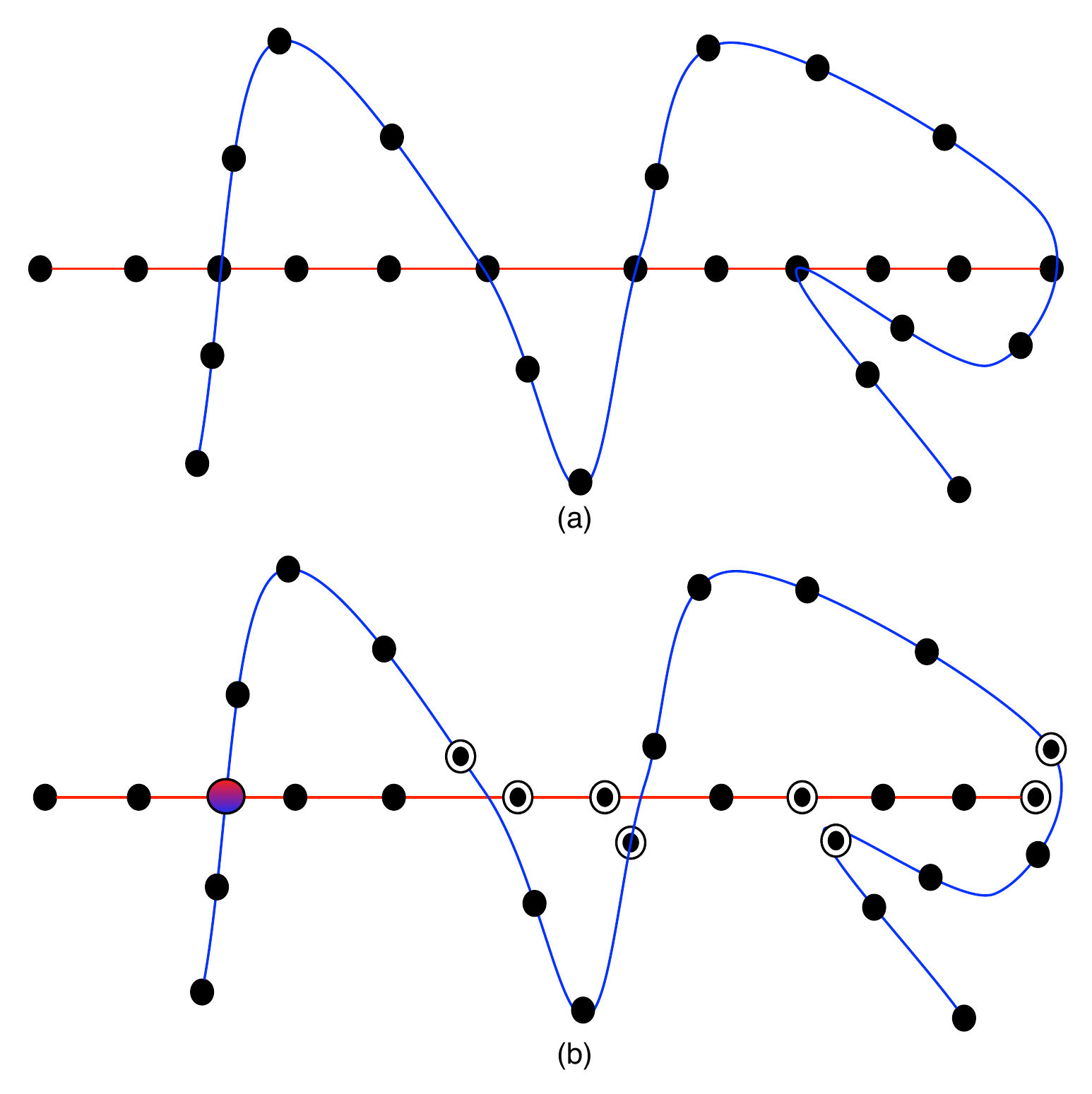}
\caption{Unravelling two paths with five common vertices (a) to obtain two paths with one common vertex (b).}
\label{fig-unravel}
\end{figure}

%We are now ready to give the reduction rule which is used to bound the number
%of vertices in maximal degree-two path in $G_i \setminus S_i$ and hence bounding the number of vertices in $T^i_{2}$.

\vspace*{2mm}
\noindent
\textsc{Reduction \alphafvs.R8.} For a path
$P^i_j\in \mathcal{P}_i$, $\tau \in \mathfrak{S}$, if $\lvert$\emph{Intercept}$(P^i_j,\tau)\rvert > 1$,
then for a vertex $u\in$ \emph{Intercept}$(P^i_j,\tau)$, \emph{Unravel}$(P^i_j,\tau,u)$.
\vspace*{2mm}

\begin{lemma}
Reduction rule \alphafvs.R8 is safe.
\label{unravel-safe}
\end{lemma}

\begin{proof}
Let $G$ be an $\alpha$-colored graph and $S_i$ be a $2$-approximate
feedback vertex set in $G_i$, for $i \in \{1,2,\dots, \alpha \}$. Let $\mathcal{P}_i$ be
the set of maximal degree-two paths in
$G_i \setminus S_i$ and $\mathfrak{S} = \mathcal{P}_1 \times \mathcal{P}_2 \times \dots \times \mathcal{P}_{\alpha}$.
For a path $P^i_j\in \mathcal{P}_i$, $\tau \in \mathfrak{S}$, $\lvert$Intercept$(P^i_j,\tau)\rvert > 1$, and $u\in $
Intercept$(P^i_j,\tau)$, let $G'$ be the $\alpha$-colored graph obtained after applying Unravel$(P^i_j,\tau,u)$ in $G$.
We show that $G$ has an \alphafvssolshort~of size at most $k$, if and only if $G'$ has an \alphafvssolshort~of size at most $k$.

In the forward direction, consider an \alphafvssolshort~$S$ in $G$ of size at
most $k$. Let $x$ be a vertex in \emph{Intercept}$(P^i_j,\tau)\setminus \{u\}$.  
We define $S' = S$ if $u \not\in S$ and $S' = (S \setminus \{u\})\cup \{x\}$ otherwise.   
A cycle $C$ in the graph $G'_t$ not containing $u_{t(\tau)}$, where $u_{t(\tau)}$ is 
the copy of $u$ created for path $t(\tau)$, $\tau \in \mathfrak{S}$, 
and $t \in \{1,2,\dots, \alpha\}$, is also a cycle in $G_t$. 
Therefore $S'$ intersects $C$. Let $P_t$ be the path in $\mathcal{P}_t$ containing 
$u$, for $t \in \{1,2,\dots,\alpha\}$. Note that in $\mathcal{P}_i$, there is exactly one 
maximal degree-two path containing $u$ and all the cycles in $G_t$ containing $u$ must contain $P_t$. 
All the cycles in $G'_t$ containing $u_{t(\tau)}$ must contain $x$, since $u_{t(\tau)}$ is the 
private copy of $u$ for the degree-two path $t(\tau)$ containing $x$.
We consider the following cases depending on whether $u$ belongs to $S$ or not.
\begin{itemize}
\item $u \in S$: A cycle $C$ in $G'_t$, $t \in \{1,2,\dots, \alpha\}$, containing $u_{t(\tau)}$ also contains $x$.
Therefore $S'$ intersects $C$. 
\item $u \notin S$: Corresponding to a cycle $C$ in $G'_t$, $t \in \{1,2,\dots, \alpha\}$,
containing $u_{t(\tau)}$, there is a cycle $C'$ on vertices $(V(C)\cup\{u\})\setminus \{u_{t(\tau)}\}$ in $G_t$.
But $S$ is an \alphafvssolshort~in $G$ and therefore both $S$ and $S'$ must contain a
vertex $y \in V(C')\setminus \{u\}$. 
\end{itemize}

In the reverse direction, let $S$ be an \alphafvssolshort~in $G'$. 
We define $S' = S$ if $\{u_{l(\tau)} \lvert u_{l(\tau)} \in S, 1 \leq l \leq \alpha\} \cap S \neq \emptyset$ 
and $S' = (S \setminus \{u_{l(\tau)} \lvert u_{l(\tau)} \in S, 1 \leq l \leq \alpha\}) \cup \{u\}$ otherwise. 
All the cycles in $G_t$ not containing $u$ are the cycles in $G'_t$ not containing $u_{t(\tau)}$. 
Therefore $S'$ intersects all those cycles. We consider the following cases depending
on whether there is some $t' \in \{1,2,\dots, \alpha\}$ for which $u_{t'(\tau)}$ belongs to $S$ or not.
\begin{itemize}
\item For all $t' \in \{1,2,\dots, \alpha\}$, $u_{t'(\tau)} \notin S$. Let $C$ be a cycle 
in $G_t$ containing $u$, for $t \in \{1,2,\dots,\alpha\}$. Note that $G'_t$ has a 
cycle $C'$ corresponding to $C$, with $V(C')=(V(C)\setminus \{u\})\cup \{u_{t(\tau)}\}$. 
$S$ intersects $C'$, therefore both $S$ and $S'$ have a vertex $y \in V(C')\setminus \{u_{t(\tau)}\}$. 
Since $y \in V(C)$, $S'$ intersects the cycle $C$ in $G_t$. 
\item For some $t' \in \{1,2,\dots, \alpha\}$, $u_{t'(\tau)} \in S$. Note that $S'$ intersects all 
the cycles in $G_t$ containing $u$, for $t \in \{1,2,\dots, \alpha\}$. Moreover, the only purpose 
of $u_{t'(\tau)}$ being in $S$ is to intersect a cycle $C'$ in $G'_t$ containing $u_{t'(\tau)}$. However, the 
corresponding cycle in $G_t$ can be intersected by a single vertex, namely $u$. Therefore, $S'$ is an \alphafvssolshort~in $G$. 
\end{itemize}

This completes the proof.
\qed
\end{proof}

\begin{theorem}
\alphafvs\ admits a kernel on $\OO(\alpha k^{3(\alpha+1)})$ vertices.
\end{theorem}

\begin{proof}
Consider an $\alpha$-colored graph $G$ on which reduction rules \alphafvs.R1 to \alphafvs.R8 have been exhaustively applied.
For $i \in \{1,2,\dots, \alpha \}$, the degree of a vertex $v \in G_i$ is either $0$ or at least $2$ in $G_i$.
Hence, in what follows, we do not count the vertices of degree $0$ in $G_i$ while counting the vertices in $G_i$;
since the total degree of a vertex $v \in V(G)$ is at least three, there
is some $j\in \{1,2,\dots,\alpha \}$ such that the degree of $v \in V(G_j)$ is at least $2$.

Let $S_i$ be a $2$-approximate feedback vertex set in $G_i$, for $i \in \{1,2,\dots, \alpha\}$.
Note that $S=\cup_{i=1}^{\alpha}S_i$ is a $2\alpha$-approximate \alphafvssolshort\ in $G$.
Let $F_i=G_i \setminus S_i$. Let $T^i_{\leq 1}$, $T^i_{2}$, and $T^i_{\geq 3}$, be the sets of
vertices in $F_i$ having degree at most one in $F_i$, degree exactly two in $F_i$,
and degree greater than two in $F_i$, respectively.

The degree of each vertex $v\in V(G_i)$ is bounded by \degbound\ in $G_i$, for $i \in \{1,2,\dots, \alpha\}$.
In particular, the degree of each $s \in S$ is bounded by \degbound\ in $G_i$.
Moreover, each vertex $v \in T^i_{\leq 1}$ has degree at least $2$ in $G_i$ and
must therefore be adjacent to some vertex in $S$. It follows that $\lvert T^i_{\leq 1} \rvert \in$ \degboundT.

In a tree, the number $t$ of vertices of degree at least three is bounded by $l-2$, where $l$ is the number of leaves.
Hence, $\lvert T^i_{\geq 3} \rvert \in$ \degboundT. Also, in a tree, the number of maximal degree-two paths is bounded by $t+l$.
Consequently, the number of degree-two paths in $G_i \setminus S_i$ is in \degboundT.
Moreover, no two maximal degree-two paths in a tree intersect.

Note that there are at most \degboundT~maximal degree-two
paths in $\mathcal{P}_i$, for $i \in\{1,2,\dots, \alpha \}$, and
therefore $\lvert \mathfrak{S} \rvert =\OO(k^{3\alpha})$. After exhaustive application
of \textsc{\alphafvs.R8}, for each path $ P^i_j \in  \mathcal{P}_i$, $i \in\{1,2,\dots, \alpha \}$,
and $\tau \in \mathfrak{S}$, there is at most one
vertex in \emph{Intercept}$(P^i_j,\tau)$. Also note that after exhaustive application
of reductions \textsc{\alphafvs.R1} to \textsc{\alphafvs.R7}, the total
degree of a vertex in $G$ is at least $3$. Therefore, there can be at most
$\OO(k^{3\alpha})$ vertices in a degree-two path $P^i_j \in \mathcal{P}_i$.
Furthermore, there are at most $\OO(k^3)$ degree-two maximal paths in
$G_i$, for $i \in \{1,2,\dots, \alpha\}$. It follows that
$\lvert T^i_2 \rvert \in \OO(k^{3(\alpha+1)})$ and
$\lvert V(G_i)\rvert \leq  \lvert T^i_{\leq 1} \rvert+\lvert T^i_2 \rvert+\lvert T^i_{\geq 3} \rvert + \lvert S_i \rvert=\OO(k^3)+\OO(k^{3(\alpha+1)})+\OO(k^3)+2k \in \OO(k^{3(\alpha+1)})$.
Therefore, the number of vertices in $G$ is in $\OO(\alpha k^{3(\alpha+1)})$.
\qed
\end{proof}

\section{Hardness results}\label{sec-hardness}
In this section we show that \logfvs, where $n$ is the
number of vertices in the input graph, is \WO-hard. We give a reduction from a special version of
the {\sc Hitting Set (HS)} problem, which we denote by {\sc $\alpha$-Partitioned Hitting Set ($\alpha$-PHS)}.
We believe this version of {\sc Hitting Set} to be of independent interest with possible applications for showing
hardness results of similar flavor. We prove \WO-hardness of
{\sc $\alpha$-Partitioned Hitting Set} by a reduction
from a restricted version of the {\sc Partitioned Subgraph Isomorphism (PSI)} problem.

Before we delve into the details, we start with a simpler reduction from {\sc Hitting Set}
showing that \nfvs\ is \WT-hard.
The reduction closely follows that of Lokshtanov~\cite{wheelhard}
for dealing with the {\sc Wheel-Free Deletion} problem. Intuitively, starting with
an instance $(\mathcal{U}, \mathcal{F}, k)$ of {\sc HS}, we first
construct a graph $G$ on $2|\mathcal{U}||\mathcal{F}|$ vertices consisting of $|\mathcal{F}|$ vertex-disjoint cycles.
Then, we use $|\mathcal{F}|$ colors to uniquely map
each set to a separate cycle; carefully connecting these cycles together
guarantees equivalence of both instances.

%\parnamedefn{{\sc Hitting Set}}
%{A tuple $(\mathcal{U}, \mathcal{F}, k)$, where $\mathcal{F}$ is a collection of subsets
%of the finite universe $\mathcal{U}$ and $k$ is a positive integer.}
%{$k$}
%{Is there a subset $X$ of $\mathcal{U}$ of cardinality at most $k$ such that
%for every $f \in \mathcal{F}$, $f \cap X$ is nonempty?}

\begin{theorem}\label{th-nfvs-hard}
\nfvs\ parameterized by solution size is \WT-hard.
\end{theorem}

\begin{proof}
Given an instance $(\mathcal{U}, \mathcal{F}, k)$ of {\sc Hitting Set}, we let $\mathcal{U} = \{u_1, \ldots, u_{|\mathcal{U}|}\}$
and $\mathcal{F} = \{f_1, \ldots, f_{|\mathcal{F}|}\}$.
We assume, without loss of generality, that each element in $\mathcal{U}$ belongs to at least one set in $\mathcal{F}$.
For each $f_i \in \mathcal{F}$, $1 \leq i \leq |\mathcal{F}|$,
we create a vertex-disjoint cycle $C_i$ on $2|\mathcal{U}|$ vertices and assign all its edges color $i$.
We let $V(C_i) = \{c^i_1, c^i_2, \ldots, c^i_{2|\mathcal{U}|}\}$ and we define
$\beta(i, u_j) = c^i_{2j - 1}$, for $1 \leq i \leq |\mathcal{F}|$ and $1 \leq j \leq |\mathcal{U}|$.
In other words, every odd-numbered vertex of $C_i$ is mapped to an element in $\mathcal{U}$.
Now for every element $u_j \in \mathcal{U}$, $1 \leq j \leq |\mathcal{U}|$,
we create a vertex $v_j$, we let $\gamma(u_j) = \{c^i_{2j - 1} | 1 \leq i \leq |\mathcal{F}| \wedge u_j \in f_i\}$, and we
add an edge (of some special color, say zero) between $v_j$ and every vertex in $\gamma(u_j)$ (see Figure~\ref{fig-hardness1}). 

\begin{figure}
\centering
\includegraphics[scale=0.5]{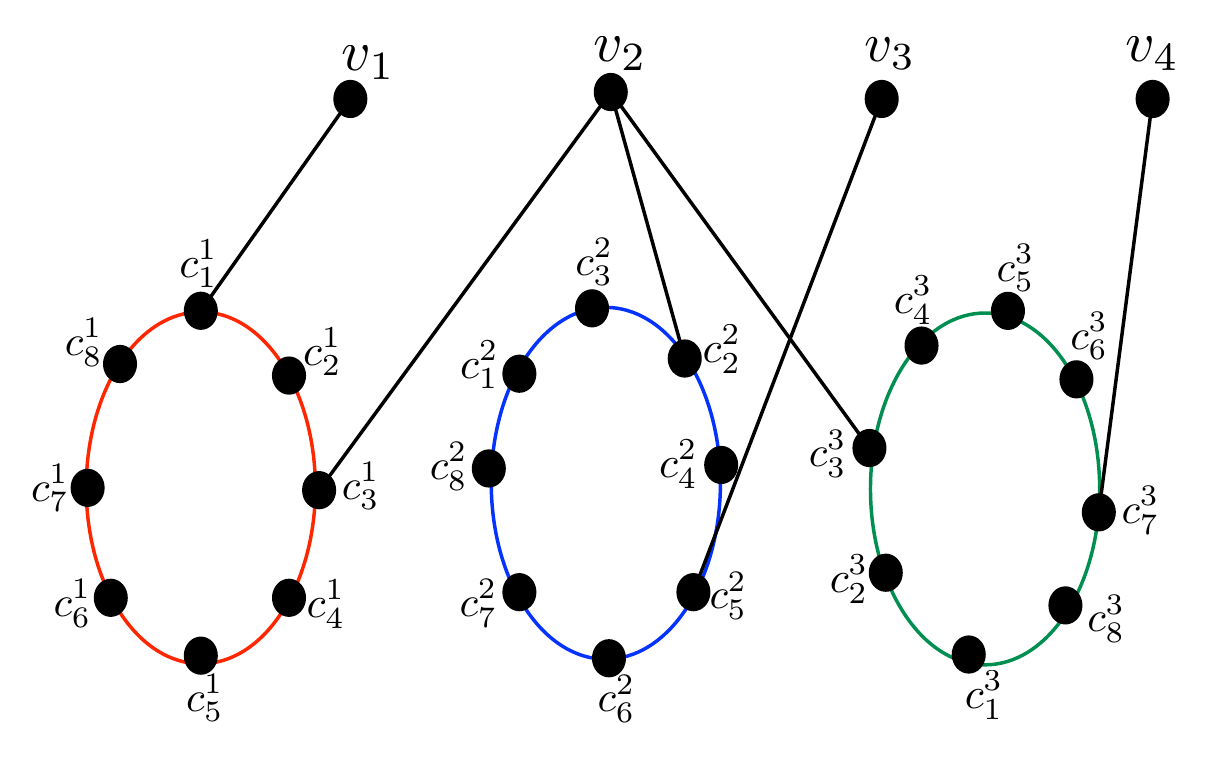}
\caption{The graph $G$ before contracting all edges colored zero for $\mathcal{U} = \{u_1, u_2, u_3, u_4\}$
and $\mathcal{F} = \{\{u_1,u_2\}, \{u_2,u_3\}, \{u_2,u_4\}\}$.}
\label{fig-hardness1}
\end{figure}

To finalize the reduction, we contract all the edges colored zero to obtain an instance $(G, k)$
of \nfvs. Note that $|V(G)| = |E(G)| = 2|\mathcal{U}||\mathcal{F}|$
and the total number of used colors is $|\mathcal{F}|$. Moreover, after contracting all $0$-colored edges,
$|\gamma(u_j)| = 1$ for all $u_j \in \mathcal{U}$.

\begin{claim}
If $\mathcal{F}$ admits a hitting set of size at most $k$ then
$G$ admits an $|\mathcal{F}|$-\fvssolshort\ of size at most $k$.
\end{claim}

\begin{proof}
Let $X = \{u_{i_1}, \ldots, u_{i_k}\}$ be such a hitting set.
We construct a vertex set $Y = \{$$\gamma(u_{i_1})$, $\ldots$, $\gamma(u_{i_k})\}$.
If $Y$ is not an $|\mathcal{F}|$-\fvssolshort\ of $G$ then $G[V(G) \setminus Y]$
must contain some cycle where all edges are assigned the same color.
By construction, every set in $\mathcal{F}$ corresponds
to a uniquely colored cycle in $G$. Hence, the contraction operations applied to
obtain $G$ cannot create new monochromatic cycles, i.e. every cycle in $G$ which does not correspond
to a set from $\mathcal{F}$ must include edges of at least two different colors.
Therefore, if $G[V(G) \setminus Y]$ contains some monochromatic cycle then $X$ cannot be
a hitting set of $\mathcal{F}$.
\qed
\end{proof}

\begin{claim}
If $G$ admits an $|\mathcal{F}|$-\fvssolshort\ of size at most $k$
then $\mathcal{F}$ admits a hitting set of size at most $k$.
\end{claim}

\begin{proof}
Let $X = \{v_{i_1}, \ldots, v_{i_k}\}$ be such an $|\mathcal{F}|$-\fvssolshort.
First, note that if some vertex in $X$ does not correspond to an element in $\mathcal{U}$,
then we can safely replace that vertex with one that does (since any such vertex belongs
to exactly one monochromatic cycle).
We construct a set $Y = \{u_{i_1}, \ldots, u_{i_k}\}$.
If there exists a set $f_i \in \mathcal{F}$ such that $Y \cap f_i = \emptyset$
then, by construction, there exists an $i$-colored cycle $C_i$ in $G$
such that $X \cap V(C_i) = \emptyset$, a contradiction.
\qed
\end{proof}

Combining the previous two claims with the fact that our reduction runs in
time polynomial in $|\mathcal{U}|$, $|\mathcal{F}|$, and $k$, completes the proof of the lemma.
\qed
\end{proof}

Notice that if we assume that $|\mathcal{U}|$ and $|\mathcal{F}|$ are linearly dependent, then
Theorem~\ref{th-nfvs-hard} in fact shows that {\sc $\OO(\sqrt{n})$-SimFVS} is \WT-hard.
However, The proof of Theorem~\ref{th-nfvs-hard} crucially relies on the fact that each cycle is
``uniquely identified'' by a separate color. In order to get around this limitation and prove
\WO-hardness of \logfvs\ we need, in some sense, to group separate sets of a
{\sc Hitting Set} instance into $\OO(\log (|\mathcal{U}||\mathcal{F}|))$
families such that sets inside each family are pairwise disjoint.
By doing so, we can modify the proof of Theorem~\ref{th-nfvs-hard} to identify all sets inside a family
using the same color, for a total of $\OO(\log n)$ colors (instead of $\OO(n)$ or $\OO(\sqrt{n})$).
We achieve exactly this in what follows.
We refer the reader to the work of Impagliazzo et al.~\cite{Impagliazzo2001367,Impagliazzo2001512} for details on
the Exponential Time Hypothesis (\textsf{ETH}).

\defparproblem{{\sc $\alpha$-Partitioned Hitting Set}}
{A tuple $(\mathcal{U}, \mathcal{F} = \mathcal{F}_1 \cup \ldots \cup \mathcal{F}_{\alpha}, k)$, where $\mathcal{F}_{i}$, $1 \leq i \leq \alpha$,
is a collection of subsets of the finite universe $\mathcal{U}$ and $k$ is a positive integer. Moreover,
all the sets within a family $\mathcal{F}_{i}$, $1 \leq i \leq \alpha$, are pairwise disjoint.}
{$k$}
{Is there a subset $X$ of $\mathcal{U}$ of cardinality at most $k$ such that
for every $f \in \mathcal{F} = \mathcal{F}_1 \cup \ldots \cup \mathcal{F}_{\alpha}$, $f \cap X$ is nonempty?}

\defparproblem{{\sc Partitioned Subgraph Isomorphism}}
{A graph $H$, a graph $G$ with $V(G) = \{g_1, \ldots, g_\ell\}$, and a coloring function $col: V(H) \rightarrow [\ell]$.}
{$k = |E(G)|$}
{Is there an injection $inj: V(G) \rightarrow V(H)$ such that for every $i \in [\ell]$,
$col(inj(g_i)) = i$ and for every $(g_i,g_j) \in E(G)$, $(inj(g_i),inj(g_j)) \in E(H)$?}

\begin{theorem}[\cite{Grohe2009218,Marx07}]\label{th-psi-hard}
{\sc Partitioned Subgraph Isomorphism} parameterized by $|E(G)|$ is \WO-hard, even
when the maximum degree of the smaller graph $G$ is three. Moreover, the problem cannot be solved in time
$f(k)n^{o({k \over \log k})}$, where $f$ is an arbitrary function, $n = |V(H)|$, and $k = |E(G)|$, unless \textsf{ETH} fails.
\end{theorem}

%We need a stronger version of Theorem~\ref{th-psi-hard}, which we prove in
%Theorem~\ref{th-boundedpsi-hard}.
We make a few simplifying assumptions:
For an instance of {\sc Partitioned Subgraph Isomorphism}, we let $H_i$ denote the
subgraph of $H$ induced on vertices colored $i$. We assume that $|H_i| = 2^t$,
for $1 \leq i \leq \ell$ and $t$ some positive integer; adding isolated vertices to each
set is enough to guarantee this size constraint. Moreover, we assume $G$ is connected
and whenever there is no edge $(g_i,g_j) \in E(G)$, then there are no
edges between $V(H_i)$ and $V(H_j)$ in $H$ (see Figure~\ref{fig-psi-instance} for an example of an instance). 
Since the {\sc PSI} problem asks for
a ``colorful'' subgraph of $H$ isomorphic to $G$ such that one vertex from $H_i$
is mapped to the vertex $g_i$, $1 \leq i \leq \ell$, it is also safe to assume
that $H_i$, $1 \leq i \leq \ell$, is edgeless.

\begin{figure}
\centering
\includegraphics[scale=0.85]{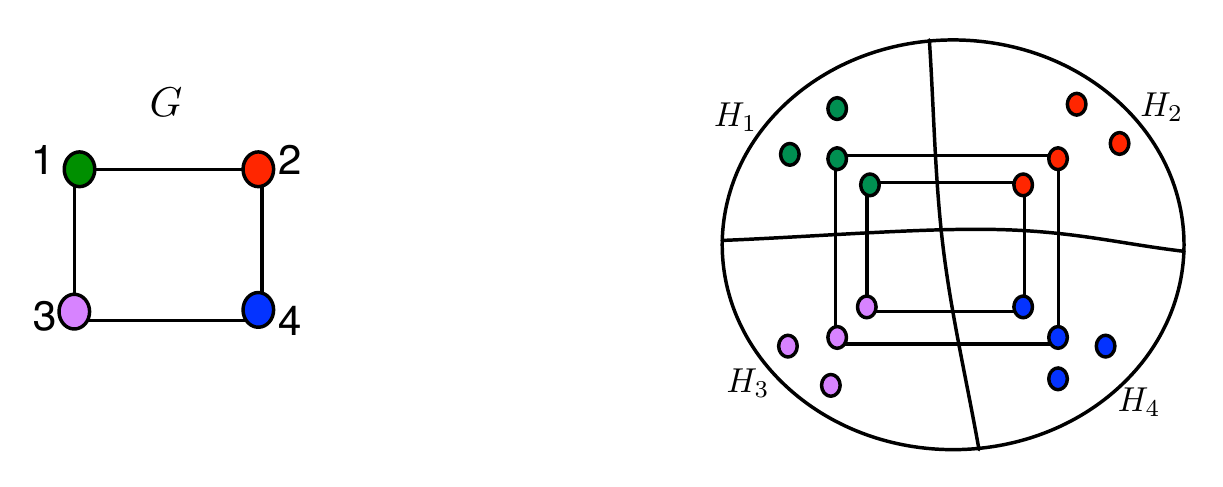}
\caption{An instance of the {\sc PSI} problem.}
\label{fig-psi-instance}
\end{figure}

\begin{theorem}\label{th-lognhs-hard}
{\sc $\OO(\log (|\mathcal{U}||\mathcal{F}|))$-Partitioned Hitting Set} parameterized by solution size is \WO-hard.
Moreover, the problem cannot be solved in time $f(k)n^{o({k \over \log k})}$, where $f$ is
an arbitrary function, $n = |\mathcal{U}|$, and $k$ is the required solution size, unless \textsf{ETH} fails.
\end{theorem}

\begin{proof}
Given an instance $(H, G, col, \ell = |V(G)|, k = |E(G)|)$ of {\sc PSI}, where $G$ has maximum degree three,
we reduce it into an instance $(\mathcal{U}, \mathcal{F} = \mathcal{F}_1 \cup \ldots \cup \mathcal{F}_{\alpha}, k' = k + \ell)$
of {\sc $\alpha$-PHS}, where $\alpha = 16 \log 2^t + 1 = 16t + 1$, $\mathcal{F}_{i}$, $1 \leq i \leq \alpha$,
is a collection of subsets of the finite universe $\mathcal{U}$, and all the
sets within a family $\mathcal{F}_{i}$ are pairwise disjoint.

We start by constructing the universe $\mathcal{U}$.
For each vertex $h^i_j \in V(H_i)$, $1 \leq i \leq \ell$ and $0 \leq j \leq 2^t - 1$, we create
an element $v^i_j$. For each edge $(h^{i_1}_{j_1}, h^{i_2}_{j_2}) \in E(H)$, we create an
element $e^{i_1,i_2}_{j_1,j_2}$ where $j_1$ is the index of the vertex in $H_{i_1}$,
$j_2$ is the index of the vertex in $H_{i_2}$, $1 \leq i_1,i_2 \leq \ell$, and $0 \leq j_1,j_2 \leq 2^t - 1$.
Note that $|\mathcal{U}| = |V(H)| + |E(H)| = \ell 2^t + |E(H)| < 4^t2\ell^2$.

We now create ``selector gadgets'' between elements corresponding to vertices and elements corresponding to edges.
For every {\em ordered} pair $(x, y)$, $1 \leq x,y \leq \ell$, such that
there exists an edge between $H_x$ and $H_y$ in $H$ (or equivalently there exists an edge $(g_x, g_y)$ in $G$),
we create $2t$ sets. We denote half of those sets by $U_{x,y,p}$ and the order half by
$D_{x,y,p}$, where $1 \leq p \leq t$. Let $\mathcal{U}_x$ denote the set of all
elements corresponding to vertices in $H_x$ and let $\mathcal{U}_{x,y}$ ($x$ and $y$ unordered in $\mathcal{U}_{x,y}$) denote the set
of all elements corresponding to edges between vertices in $H_x$ and vertices $H_y$.
We let $bit(i)[p]$, $0 \leq i \leq 2^t - 1$ and $1 \leq p \leq t$,
be the $p^{th}$ bit in the bit representation of $i$ (where position $p = 1$ holds the most significant bit).
For each $v^x_i \in \mathcal{U}_x$ and for all $p$ from $1$ to $t$,
if $bit(i)[p] = 0$ we add $v^x_i$ to set $D_{x,y,p}$
and we add $v^x_i$ to set $U_{x,y,p}$ otherwise.
For each $e^{x,y}_{i,j} \in \mathcal{U}_{x,y}$ and for all $p$ from $1$ to $t$,
if $bit(i)[p] = 0$ we add $e^{x,y}_{i,j}$ to set $U_{x,y,p}$
and we add $e^{x,y}_{i,j}$ to set $D_{x,y,p}$ otherwise. Recall that for $e^{x,y}_{i,j}$, $i$
corresponds to the index of element $v^x_i \in \mathcal{U}_x$.

\begin{figure}
\centering
\includegraphics[scale=0.6]{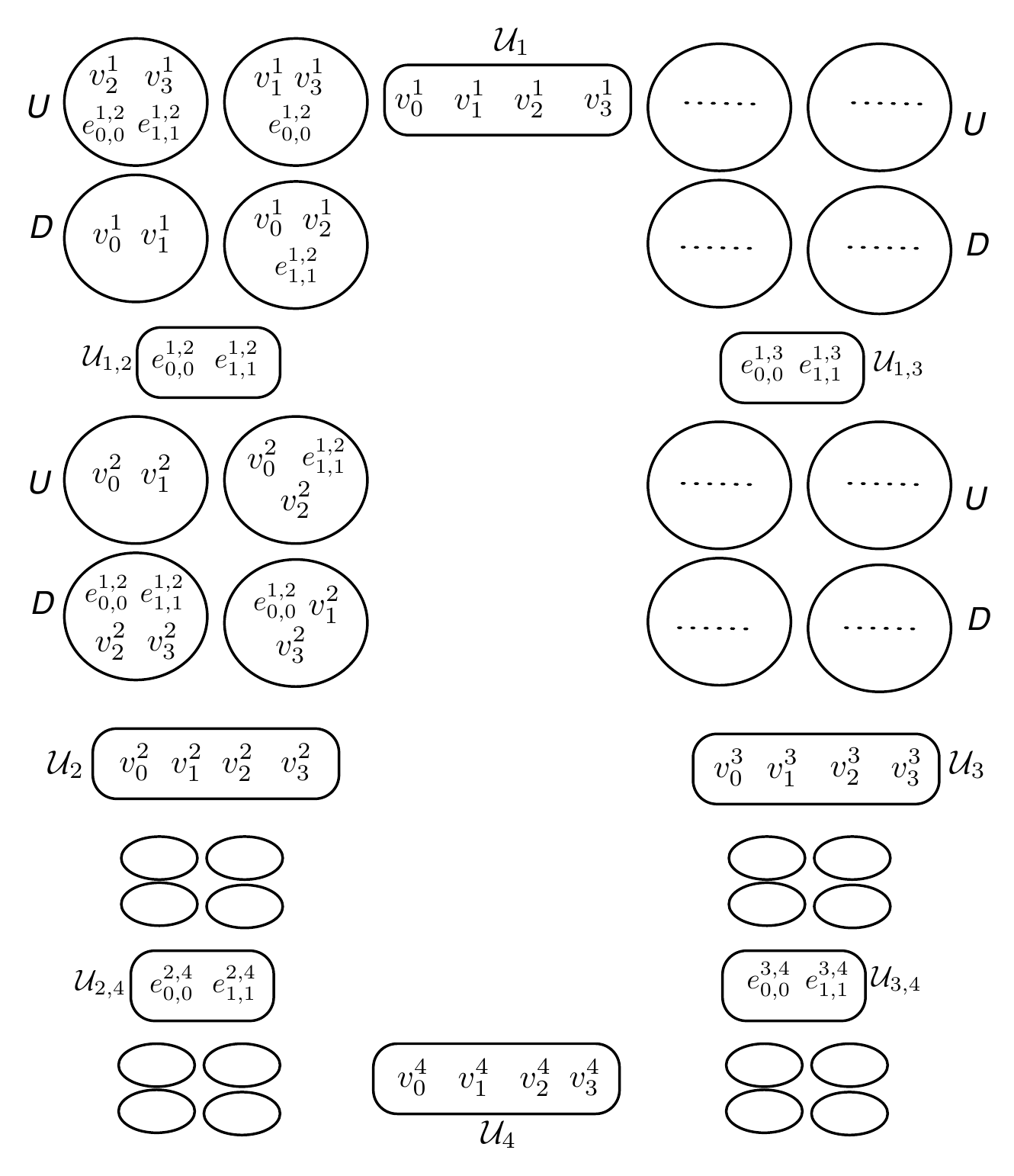}
\caption{Parts of the reduction for the {\sc PSI} instance from Figure~\ref{fig-psi-instance}. 
Rectangles represents subsets of the universe and circles represent sets in the family. }
\label{fig-hardness2}
\end{figure}

Finally, for each $x$, $1 \leq x \leq \ell$, we add a set $Q_x = \mathcal{U}_x$, and
for each (unordered) pair $x$,$y$ such that $(g_x, g_y) \in E(G)$ we add a set $Q_{x,y} = \mathcal{U}_{x,y}$.
Put differently, a set $Q_x$ contains all elements corresponding to vertices in $H_x$ and
a set $Q_{x,y}$ contains all elements corresponding to edges between $H_x$ and $H_y$.
The role of these $\ell + k$ sets is simply to force a solution to pick at least one
element from every $\mathcal{U}_x$ and one element from every $\mathcal{U}_{y,z}$, $1 \leq x,y,z \leq \ell$.
Note that we have a total of $4t|E(G)| + |E(G)| + \ell < 4t\ell^2 + \ell^2 + \ell$ sets and
therefore $16t + 1 \in \OO(\log (|\mathcal{U}||\mathcal{F}|))$.
We set $k' = |V(G)| + |E(G)| = \ell + k$. This completes the construction. 
An example of the construction for the instance given in Figure~\ref{fig-psi-instance} is provided in Figure~\ref{fig-hardness2}.

\begin{claim}
In the resulting instance $(\mathcal{U}, \mathcal{F} = \mathcal{F}_1 \cup \ldots \cup \mathcal{F}_{\alpha}, k' = k + \ell)$,
$\alpha = 16 \log 2^t + 1= 16t + 1$.
\end{claim}

\begin{proof}
First, we note that all sets $Q_x$ and $Q_{y,z}$, $1 \leq x,y,z \leq \ell$, are pairwise disjoint. Hence,
we can group all these sets into a single partition.
We now prove that $16t$ is enough to partition the remaining sets.

Since $G$ has maximum degree three, we know by Vizing's theorem~\cite{MR0180505} that
$G$ admits a proper $4$-edge-coloring, i.e. no two edges incident on the same vertex receive
the same color. Let us fix such a $4$-edge-coloring and denote it by $\beta : E(G) \rightarrow \{1, 2, 3, 4\}$.
Recall that for every ordered pair $(x, y)$, $1 \leq x,y \leq \ell$,
we define two groups of sets $U_{x,y,p}$ and $D_{x,y,p}$, $1 \leq p \leq t$.
Given any set $X_{x,y,p}$, $X \in \{U,D\}$, we define the partition to which $X_{x,y,p}$ belongs as
$part(X, x, y, p) = (\beta(g_x,g_y), p, \{U,D\}, \{x<y, x>y\})$. In other words,
we have a total of $16t$ partitions depending on the color of the edge $(g_x,g_y)$ in $G$,
the position $p$, whether $X = U$ or $X = D$, and whether $x < y$ or $x > y$
(recall that we assume $x \neq y$).

Since $\beta$ is a proper $4$-coloring of the edges of $G$, we know that if
two sets belong to the same partition they must be of the form
$X_{x_1,y_1,p}$ and $X_{x_2,y_2,p}$, where $X \in \{U,D\}$,
$x_1 \neq x_2$, $y_1 \neq y_2$, $\beta(g_{x_1},g_{y_1}) = \beta(g_{x_2},g_{y_2})$,
$x_1 < y_1$ ($x_1 > y_1$), and $x_2 < y_2$ ($x_2 > y_2$).
It follows from our construction that $X_{x_1,y_1,p} \cap X_{x_2,y_2,p} = \emptyset$;
$X_{x_1,y_2,p}$ only contains elements from $\mathcal{U}_{x_1} \cup \mathcal{U}_{x_1,y_1}$,
$X_{x_2,y_2,p}$ only contains elements from $\mathcal{U}_{x_2} \cup \mathcal{U}_{x_2,y_2}$,
and $(\mathcal{U}_{x_1} \cup \mathcal{U}_{x_1,y_1})$ $\cap$ $(\mathcal{U}_{x_2} \cup \mathcal{U}_{x_2,y_2})$ is empty.
\qed
\end{proof}

\begin{claim}
The resulting instance $(\mathcal{U}, \mathcal{F} = \mathcal{F}_1 \cup \ldots \cup \mathcal{F}_{\alpha}, k' = k + \ell)$
admits no hitting set of size $k' - 1$.
\end{claim}

\begin{proof}
If there exists a hitting set $S$ of size $k' - 1$, then either (1) there exists $\mathcal{U}_x$, $1 \leq x \leq \ell$,
such that $S \cap \mathcal{U}_x = \emptyset$ or (2) there exists $\mathcal{U}_{y,z}$, $1 \leq y,z \leq \ell$,
such that $S \cap \mathcal{U}_{y,z} = \emptyset$. In case (1), we are left with a set $Q_x$ which
is not hit by $S$. Similarly, for case (2), there exists a set $Q_{y,z}$ which is not hit by $S$.
In both cases we get a contradiction as we assumed $S$ to be a hitting set, as needed.
\qed
\end{proof}

\begin{claim}
Any hitting set of size $k'$ of the resulting instance
$(\mathcal{U}, \mathcal{F} = \mathcal{F}_1 \cup \ldots \cup \mathcal{F}_{\alpha}, k' = k + \ell)$ must pick exactly one element
from each set $\mathcal{U}_x$, $1 \leq x \leq \ell$, and exactly one element
from each set $\mathcal{U}_{y,z}$, $1 \leq y,z \leq \ell$. Moreover,
for every ordered pair $(x,y)$, $1 \leq x,y \leq \ell$, a hitting set of size $k'$ must pick
$v^x_i \in \mathcal{U}_x$ and $e^{x,y}_{i,j} \in \mathcal{U}_{x,y}$, $0 \leq i,j \leq 2^t - 1$. In other words,
the vertex $h^x_i \in V(H)$ is incident to the edge $(h^x_i, h^y_j) \in E(H)$.
\end{claim}

\begin{proof}
The first part of the claim follows from the previous claim combined with the fact
that $k' = k + \ell$. For the second part,
assume that there exists a hitting set $S$ of size $k'$ such that
for some ordered pair, $(x,y)$, $S$ includes
$v^x_{i_1} \in \mathcal{U}_x$ and $e^{x,y}_{i_2,j} \in \mathcal{U}_{x,y}$, where $i_1 \neq i_2$.
Since $i_1 \neq i_2$, then $bit(i_1)[p] \neq bit(i_2)[p]$ for at least one position $p$.
For that position, we know that $v^x_{i_1}$ and $e^{x,y}_{i_2,j}$ must both belong
to only one of $U_{x,y,p}$ or $D_{x,y,p}$.
Hence, either $U_{x,y,p}$ or $D_{x,y,p}$ is not hit by $v^x_{i_1}$ and $e^{x,y}_{i_2,j}$ when $i_1 \neq i_2$.
\qed
\end{proof}

\begin{claim}
If $(H, G, col, \ell = |V(G)|, k = |E(G)|)$, where $G$ has maximum degree three,
is a yes-instance of {\sc PSI} then $(\mathcal{U}, \mathcal{F} = \mathcal{F}_1 \cup \ldots \cup \mathcal{F}_{\alpha}, k' = k + \ell)$
is a yes-instance of {\sc $\alpha$-PHS}.
\end{claim}

\begin{proof}
Let $S$, a subgraph of $H$, denote the solution graph and let
$V(S) = \{h^1_{i_1}, \ldots, h^\ell_{i_\ell}\}$. We claim that
$S' = \{v^1_{i_1}, \ldots, v^\ell_{i_\ell}\}$ $\cup$
$\{e^{x, y}_{j_1,j_2} | (g_{x},g_{y}) \in E(G) \wedge j_1,j_2 \in \{i_1, \ldots, i_\ell\}\}$ is a hitting set of $\mathcal{F}$.
That is, the hitting set picks $\ell$ elements corresponding
to the $\ell$ vertices in $S$ (or $G$) and $k$ elements corresponding
to the $k$ edges in $G$.

Clearly, all sets $Q_x$ and $Q_{y,z}$, $1 \leq x,y,z \leq \ell$,
are hit since we pick one element from each.
We now show that all sets $U_{x,y,p}$ and $D_{x,y,p}$, $1 \leq x,y \leq \ell$ and $1 \leq p \leq t$, are also hit.
Assume, without loss of generality, that for fixed $x$, $y$, and $p$, some set $U_{x,y,p}$ is not hit.
Let $v^x_{i_1} \in \mathcal{U}_x$ be the element we picked from $\mathcal{U}_x$ and let
$e^{x,y}_{i_2,j}$ be the element we picked from $\mathcal{U}_{x,y}$.
If $U_{x,y,p}$ is not hit, it must be the case that $i_1 \neq i_2$
which, by the previous claim, is not possible.
\qed
\end{proof}

\begin{claim}
If $(\mathcal{U}, \mathcal{F} = \mathcal{F}_1 \cup \ldots \cup \mathcal{F}_{\alpha}, k' = k + \ell)$
is a yes-instance of {\sc $\alpha$-PHS} then $(H, G, col, \ell = |V(G)|, k = |E(G)|)$ is a yes-instance of {\sc PSI}.
\end{claim}

\begin{proof}
Let $S = \{v^1_{i_1}, \ldots, v^\ell_{i_\ell}\}$ $\cup$
$\{e^{x, y}_{j_1,j_2} | (g_{x},g_{y}) \in E(G) \wedge j_1,j_2 \in \{i_1, \ldots, i_\ell\}\}$ be a hitting set of $\mathcal{F}$.
Note that we can safely assume that the hitting set picks such
elements since it has to hit all sets $Q_x$ and $Q_{y,z}$, $1 \leq x,y,z \leq \ell$.
We claim that the subgraph $S'$ of $H$ with vertex set $V(S') = \{h^1_{i_1}, \ldots, h^\ell_{i_\ell}\}$ is a solution
to the {\sc PSI} instance.

By construction, there is an injection $inj: V(G) \rightarrow V(S')$ such that for every $i \in [\ell]$,
$col(inj(g_i)) = i$. In fact, $S'$ contains exactly one vertex for each color $i \in [\ell]$.
Assume that there exists an edge $(g_i,g_j) \in E(G)$ such that $(inj(g_i),inj(g_j)) \not\in E(S')$.
This implies that there exists two vertices $h^x_{i}, h^y_{j} \in V(S')$ such that
$(h^x_{i},h^y_{j}) \not\in E(S')$. But we know that there exists at least one edge, say $(h^x_{i'}, h^y_{j'})$, between
vertices in $H_x$ and vertices in $H_y$ (from our assumptions). Since $i' \neq i$, $j' \neq j$,
$v^x_{i}, v^y_{j} \in S$, and $e^{x, y}_{i,j} \not\in S$, it follows that $S$ cannot be a hitting
set of $\mathcal{F}$ as at least one set in $U_{x,y,p} \cup D_{x,y,p}$ and one set in $U_{y,x,p} \cup D_{y,x,p}$
is not hit by $S$, a contradiction.
\qed
\end{proof}

This completes the proof of the theorem.
\qed
\end{proof}

We are now ready to state the main result of this section. The proof
of Theorem~\ref{th-lognfvs-hard} follows the same steps as the
proof of Theorem~\ref{th-nfvs-hard} with one exception, i.e we reduce
from {\sc $\OO(\log (|\mathcal{U}||\mathcal{F}|))$-Partitioned Hitting Set} and use
$\OO(\log (|\mathcal{U}||\mathcal{F}|))$ colors instead of $|\mathcal{F}|$.

\begin{theorem}\label{th-lognfvs-hard}
\logfvs\ parameterized by solution size is \WO-hard.
\end{theorem}

\begin{proof}
Given an instance $(\mathcal{U}, \mathcal{F} = \mathcal{F}_1 \cup \ldots \cup \mathcal{F}_{\alpha}, k)$
of {\sc $\alpha$-PHS}, we let $\mathcal{U} = \{u_1, \ldots, u_{|\mathcal{U}|}\}$
and $\mathcal{F}_i = \{f^i_1, \ldots, f^i_{|\mathcal{F}_i|}\}$, $1 \leq i \leq \alpha$.
We assume, without loss of generality, that each element in
$\mathcal{U}$ belongs to at least one set in $\mathcal{F}$.

For each $f^i_j \in \mathcal{F}_i$, $1 \leq i \leq \alpha$ and $1 \leq j \leq |\mathcal{F}_i|$,
we create a vertex-disjoint cycle $C^i_j$ on $2|\mathcal{U}|$ vertices and assign all its edges color $i$.
We let $V(C^i_j) = \{c^{i,j}_1, \ldots, c^{i,j}_{2|\mathcal{U}|}\}$ and we define
$\beta(i, j, u_p) = c^{i,j}_{2p - 1}$, $1 \leq i \leq \alpha$, $1 \leq j \leq |\mathcal{F}_i|$, and $1 \leq p \leq |\mathcal{U}|$.
In other words, every odd-numbered vertex of $C^i_j$ is mapped to an element in $\mathcal{U}$.
Now for every element $u_p \in \mathcal{U}$, $1 \leq p \leq |\mathcal{U}|$,
we create a vertex $v_p$, we let
$\gamma(u_p) = \{c^{i,j}_{2p - 1} | 1 \leq i \leq \alpha \wedge 1 \leq j \leq |\mathcal{F}_i| \wedge u_p \in f^i_j\}$, and we
add an edge (of some special color, say $0$) between $v_p$ and every vertex in $\gamma(u_p)$.
To finalize the reduction, we contract all the edges colored $0$ to obtain an instance $(G, k)$
of \logfvs. Note that $|V(G)| = |E(G)| = 2|\mathcal{U}||\mathcal{F}|$
and the total number of used colors is $\alpha$. Moreover, after contracting all special edges,
$|\gamma(u_p)| = 1$ for all $u_p \in \mathcal{U}$.

\begin{claim}
If $\mathcal{F}$ admits a hitting set of size at most $k$ then
$G$ admits an $\alpha$-\fvssolshort\ of size at most $k$.
\end{claim}

\begin{proof}
Let $X = \{u_{p_1}, \ldots, u_{p_k}\}$ be such a hitting set.
We construct a vertex set $Y = \{$$\gamma(u_{p_1})$, $\ldots$, $\gamma(u_{p_k})\}$.
If $Y$ is not an $\alpha$-\fvssolshort\ of $G$ then $G[V(G) \setminus Y]$
must contain some monochromatic cycle.
By construction, only sets from the same family $\mathcal{F}_i$, $1 \leq i \leq \alpha$, correspond
to cycles assigned the same color in $G$. But since we started with an instance of
{\sc $\alpha$-PHS}, no two such sets intersect. Hence, the contraction operations applied to
obtain $G$ cannot create new monochromatic cycles.
Therefore, if $G[V(G) \setminus Y]$ contains some monochromatic cycle then $X$ cannot be
a hitting set of $\mathcal{F}$.
\qed
\end{proof}

\begin{claim}
If $G$ admits an $\alpha$-\fvssolshort\ of size at most $k$
then $\mathcal{F}$ admits a hitting set of size at most $k$.
\end{claim}

\begin{proof}
Let $X = \{v_{p_1}, \ldots, v_{p_k}\}$ be such an $\alpha$-\fvssolshort.
First, note that if some vertex in $X$ does not correspond to an element in $\mathcal{U}$,
then we can safely replace that vertex with one that does (since any such vertex belongs
to exactly one monochromatic cycle).
We construct a set $Y = \{u_{p_1}, \ldots, u_{p_k}\}$.
If there exists a set $f^i_j \in \mathcal{F}_i$ such that $Y \cap f^i_j = \emptyset$
then, by construction, there exists an $i$-colored cycle $C_i$ in $G$
such that $X \cap V(C_i) = \emptyset$, a contradiction.
\qed
\end{proof}

Combining the previous two claims with the fact that our reduction runs in
time polynomial in $|\mathcal{U}|$, $|\mathcal{F}|$, and $k$, completes the proof of the theorem.
\qed
\end{proof}

\section{Conclusion}\label{conclusion}
We have showed that \alphafvs\ parameterized by solution size $k$ is fixed-parameter tractable and
can be solved by an algorithm running in $\OO^\star(23^{\alpha k})$ time, for any constant $\alpha$.
For the special case of $\alpha=2$, we gave a faster $\OO^\star(81^{k})$ time algorithm which follows
from the observation that the base case of the general algorithm can be solved in polynomial
time when $\alpha=2$. Moreover, for constant $\alpha$, we presented
a kernel for the problem with $\OO(\alpha k^{3(\alpha + 1)})$ vertices.

It is interesting to note that our algorithm implies that \alphafvs\ can be solved
in $(2^{\OO(\alpha )})^k n^{\OO(1)}$ time. However,
we have also seen that \alphafvs\ becomes \WO-hard when $\alpha \in \OO(\log n)$.
This implies that (under plausible complexity assumptions)
an algorithm running in $(2^{o(\alpha )})^k n^{\OO(1)}$ time cannot exist.
In other words, the running time cannot be subexponential in either $k$ or $\alpha$.

As mentioned by Cai and Ye~\cite{caiye2014}, we believe that studying
generalizations of other classical problems to edge-colored graphs is well motivated
and might lead to interesting new insights about combinatorial and structural properties
of such problems. Some of the potential candidates are
{\sc Vertex Planarization}, {\sc Odd Cycle Transversal}, {\sc Interval Vertex Deletion},
{\sc Chordal Vertex Deletion}, \pfd, and, more generally, \alphahdel.

\bibliographystyle{abbrv}
\bibliography{reference}
\end{document}